\documentclass[10pt,journal,compsoc]{IEEEtran}

\ifCLASSOPTIONcompsoc
  \usepackage[nocompress]{cite}
\else
  \usepackage{cite}
\fi

\ifCLASSINFOpdf
\else
\fi

\ifCLASSOPTIONcompsoc

\else

\fi

\ifCLASSINFOpdf

\else

\fi
\usepackage{float}
\usepackage[center]{caption2}
\usepackage{xspace}
\usepackage{balance}
\usepackage{algorithm}
\usepackage{algorithmic}
\usepackage{diagbox}
\usepackage{bbding}
\usepackage{amsthm}
\usepackage{amsmath,amssymb,amsfonts}
\usepackage{makecell}
\usepackage{diagbox}
\usepackage{multirow}
\usepackage{threeparttable}
\usepackage{subfigure}
\usepackage{booktabs}
\usepackage{tikz}

\newenvironment{packeditemize}{
\begin{list}{$\bullet$}{
\setlength{\labelwidth}{8pt}
\setlength{\itemsep}{0pt}
\setlength{\leftmargin}{\labelwidth}
\addtolength{\leftmargin}{\labelsep}
\setlength{\parindent}{0pt}
\setlength{\listparindent}{\parindent}
\setlength{\parsep}{0pt}
\setlength{\topsep}{3pt}}}{\end{list}}
\newtheorem{theorem}{Theorem}[section]

\newcommand*{\circled}[1]{\lower.7ex\hbox{\tikz\draw (0pt, 0pt)
    circle (.4em) node {\makebox[0.9em][c]{\small #1}};}}
\usepackage[colorlinks,
            linkcolor=black,
            anchorcolor=black,
            citecolor=black]{hyperref}

\renewcommand{\raggedright}{\leftskip=0pt \rightskip=0pt plus 0cm}
\raggedright
\begin{document}
\title{New Secure Sparse Inner Product with Applications to Machine Learning}

\author{{Guowen~Xu, Shengmin~Xu, Jianting~Ning, Tianwei~Zhang, Xinyi~Huang, Hongwei~Li, Rongxing~Lu~\IEEEmembership{Fellow,~IEEE}}

\IEEEcompsocitemizethanks{\IEEEcompsocthanksitem Guowen~Xu and Tianwei~Zhang are with the School of Computer Science and Engineering, Nanyang Technological University. (e-mail: guowen.xu@ntu.edu.sg;  tianwei.zhang@ntu.edu.sg)

\IEEEcompsocthanksitem Xinyi~Huang is with the Artificial Intelligence Thrust, Information Hub, Hong Kong University of Science and Technology (Guangzhou), Guangzhou, China, 511458 (e-mail: xinyi@ust.hk)

\IEEEcompsocthanksitem Shengmin~Xu and Jianting~Ning are with the College of Computer and Cyber Security, Fujian Normal University, Fuzhou, China (e-mail: smxu1989@gmail.com; jtning88@gmail.com)

\IEEEcompsocthanksitem Hongwei~Li  is with the school of Computer Science and Engineering,  University of Electronic Science and Technology of China, Chengdu 611731, China.(e-mail: hongweili@uestc.edu.cn)

\IEEEcompsocthanksitem Rongxing~Lu  is with the school of  Computer Science, University of New Brunswick, Canada.(e-mail: rlu1@unb.ca)
}}

 \IEEEcompsoctitleabstractindextext{
\begin{abstract}
 \raggedright
 Sparse inner product (SIP)  has the attractive property of overhead being dominated by the intersection of inputs between parties, independent of the actual input size. It has intriguing prospects, especially for boosting machine learning on large-scale data, which is tangled with sparse data. In this paper, we investigate privacy-preserving SIP problems that have rarely been explored before. Specifically, we propose two concrete constructs, one requiring offline linear communication which can be amortized  across queries, while the other has sublinear overhead but relies on the more computationally expensive tool. Our approach exploits state-of-the-art cryptography tools including garbled Bloom filters (GBF) and Private Information Retrieval (PIR) as the cornerstone, but carefully fuses them to obtain non-trivial overhead reductions. We provide formal security analysis of the proposed constructs and implement them into representative machine learning algorithms including k-nearest neighbors, naive Bayes classification and logistic regression. Compared to the  existing efforts, our method achieves $2$-$50\times$ speedup in runtime and up to $10\times$ reduction in communication.
\end{abstract}
}
\maketitle

\IEEEdisplaynotcompsoctitleabstractindextext

\IEEEpeerreviewmaketitle

\section{Introduction}
The sparse inner product (SIP)\cite{ma2021learning,srivastava2020matraptor}, as the basis for sparse linear algebra including matrix multiplication, matrix-vector inner product, and matrix inversion, has shown an irreplaceable role  in various applications, especially in accelerating large-scale machine learning (ML) where sparsity is intertwined. Take the classification task with 20Newsgroups dataset \cite{ruff2019self}  as an example. It consists of over 9000 vectors each of which includes approximately $10^5$ dimensions; however, each vector on average contains less than 100 non-zero values (approximately $0.1\%$). Performing SIP on such a sparse dataset boosts performance by at least an order of magnitude compared to traditional dense multiplication. It stems from the fact that the complexity of the sparse operation only depends on the intersection of the number of non-zero data between the datasets, and is independent of the original data dimension. Beyond classification tasks, SIP has been widely used in various fields of machine learning such as k-nearest neighbors \cite{liao2021deep}, cluster analysis \cite{Tsiokanos22}, naive Bayes, and logistic regression \cite{Yansong21}.

 While SIP is appreciated for improving performance on sparse data, it inherits all the privacy issues that arise from plaintext computations \cite{cui2021exploiting,schoppmann2019make}. Consider a ML inference platform consisting of a client and a server. The client feeds sparse query vectors to the server (which  holds a sparse model), and then the server provides inference results about the query to the client. To facilitate SIP, the client is required to provide plaintext queries while the server needs to expose sparsity details of model parameters. It is clearly  a breach of privacy \cite{zheng2019helen,mishra2020delphi}. Concretely, client queries are often naturally sensitive and may contain personal physiological information, financial information, and disease history, depending on the application. Outsourcing these private data to untrusted third parties inevitably raises privacy concerns \cite{sav2020poseidon,juvekar2018gazelle}. Model parameters, as precious intellectual property rights, should also be reasonably protected to ensure the market competitiveness of service providers.

\subsection{Related Works}
\label{Related Works and  Challenges}
While privacy-preserving machine learning \cite{lehmkuhl2021muse,chandran2021simc,jiang2018secure} has been extensively investigated, sparse linear algebra, especially SIP and its applications in ML are rarely explored. Existing efforts suffer from either scalability (i.e. customization to specific scenarios) or inefficiency (requiring generic secure multi-party computation (MPC) protocols). Below we briefly review these works and provide a discussion of their limitations.

\textit{Chen et al}.\cite{chen2021homomorphic} design a sparse matrix multiplication by carefully combining two primitives, homomorphic encryption and secret sharing. Its core idea is to use fully homomorphic encryption (FHE) to realize the multiplication of a sparse plaintext matrix and any encrypted matrix (ignoring its sparseness) from the client. The ciphertext result is then secretly shared with parties for further computation. However, \cite{chen2021homomorphic} does not consider the sparsity of the client's input, it may not perform well compared to dense multiplication in some scenarios, especially when the data held by the server are dense while the client input is highly sparse. \textit{Cui et al}. proposed S$^3$Rec \cite{cui2021exploiting}, which improves the efficiency of \cite{chen2021homomorphic} by combining homomorphic encryption with existing Private Information Retrieval (PIR)\cite{angel2018pir} techniques. S$^3$Rec proposes two different sparse operations depending on whether sparsity is considered sensitive. When the  dataset sparsity  is accessible, S$^3$Rec relies heavily on Beaver's triples to achieve fast matrix multiplication. In contrast, S$^3$Rec uses PIR as the underlying technique to obtain the intersection between two matrices when sparsity is agnostic. Furthermore, homomorphic encryption is used to realize the multiplication between two sparse matrices. Since S$^3$Rec is dedicated to exploring data sparsity in the secure cross-platform social recommendation, it is highly scenario-specific and requires non-trivial evolution for general ML applications.

The closest work to this paper is ROOM \cite{schoppmann2019make}, which focuses on designing low-level secure linear algebra that can be applied to all ML scenarios that require linear operations.  ROOM defines a new cryptographic primitive, Read-Only-Oblivious Map, and uses it as a building block to implement other linear operations including Gather, Scatter, and sparse matrix multiplication. The fly in the ointment is that ROOM relies heavily on the general MPC protocol \cite{bellare2012foundations, huang2014amortizing} to achieve the desired secure computation. This is usually computationally time-consuming and inevitably incurs non-negligible communication rounds. Therefore, from the status quo, designing a generic sparse linear algebra for practical applications still leaves too much to be desired.
\subsection{Technical Challenges}
This paper aims to break the dilemma of previous work and provide a secure SIP approach towards practicality. In briefly, we design a highly optimized secure SIP (called S-SIP), and then we extend S-SIP to general ML scenarios to demonstrate its superiority for accelerating computation. Note that fulfilling the above aspirations is non-trivial and requires careful addressing of the following challenges.

\begin{packeditemize}
\item \textit{How to get out of the cage of inefficiency}? Existing work heavily relies on the generic MPC protocol to provide secure set intersection followed by inner product. This is clearly doable  but at the cost of incurring potentially unnecessary overhead. However, bypassing the generic MPC to design a customized secure SIP requires very careful design. A potential challenge is how to simultaneously infer the intersection of two inputs and complete the inner product. Intuitively, we can use existing techniques such as the Private Set Intersection (PSI) \cite{lepoint2021private} or PIR \cite{mahdavi2022constant} to first obtain the intersection and then use Beaver's triples \cite{keller2018overdrive} or homomorphic encryption to perform linear operations. However, PSI inherently leaks the intersection itself, which is not allowed in S-SIP and requires careful modification to accommodate higher security requirements. PIR is a promising method, but it still needs to be highly optimized such as batch query, recursion and oblivious expansion to speed up retrieval performance.

\item \textit{How to design S-SIP with satisfactory scalability}? This means that the constructed S-SIP should exhibit adaptable performance for datasets of different scales.  As in machine learning scenarios, the data held by the server may be static and small-scale, or include a large number of entries. For the former, it is desirable if there exists a way that the majority of computations are performed offline, i.e., independently of client input. It is bound to significantly speed up the computation in the online phase. For the latter, we also expect the overhead to be only linear with the size of the smaller dataset (usually the input of clients) and logarithmic with the larger dataset. However, there is no previous work to achieve the above requirements.

\item \textit{How to enable fast S-SIP without privacy trade-offs}? Existing general S-SIP methods need to expose some sparsity information to reduce computation time. For example, when performing the matrix-vector inner product, ROOM \cite{schoppmann2019make} is forced to reveal sparsity information including the number of non-zero rows or columns in the matrix and the number of non-zero entries in the vector. These messages are sometimes privacy-critical. Especially in the scenario of medical data analysis, the leakage of non-zero entries can often be used by the adversary as side information to infer whether the target user has infected certain diseases. Therefore, it is necessary to design S-SIP without any privacy trade-offs, which requires to design a new inner product operation that is fundamentally different from previous work.

\end{packeditemize}

\subsection{Our Contributions}
\label{Our Contributions}
In this paper, we propose two efficient S-SIP constructs (called S-SIP$_1$ and S-SIP$_2$) to address the above challenges. S-SIP$_1$ uses the Bloom filter (BF) and its variant, the garbled Bloom filter (GBF) \cite{dong2013private}, as the underlying building blocks. It incurs an overhead linearly proportional to the size of the server's dataset (larger) in the offline phase, but the overhead in the online phase is completely independent of the size of the server's data. This makes S-SIP$_1$ ideal for ML scenarios where the server holds small and fixed datasets. S-SIP$_2$ is fully online without precomputing. We use the state-of-the-art PIR technology \cite{mughees2021onionpir} as the underlying technology and extend it to the batch query mode to reduce overhead through amortization. S-SIP$_2$ enables the overhead of S-SIP to be linear to the client's input (smaller) and logarithmic to the size of the server's dataset. Our constructions does not require any privacy tradeoffs to gain performance benefits. We provide a formal security analysis as well as extensive experiments to demonstrate the semantic security and superiority of the proposed schemes. In summary, our contributions are as follows:
\begin{packeditemize}
\item We present a new S-SIP primitive that is general, efficient, scalable, and can be used in any linear operation where data sparsity exists.

\item We design two concrete structures, S-SIP$_1$ and S-SIP$_2$, where the former requires computationally intensive offline operations but exhibits superior online performance. The latter relies on an optimized state-of-the-art PIR technique whose overhead grows only logarithmically with the size of the dataset held by the server.

\item  We provide a formal security analysis of the proposed constructs and  implement them  into representative machine learning algorithms including k-nearest neighbors, naive Bayes and logistic regression. Compared to the  existing efforts, our method achieves $2$-$50\times$ speedup in runtime and up to $10\times$ reduction in communication.
\end{packeditemize}
\textbf{Roadmap}: The remainder of this paper is organized as follows. In  Section \ref{sec:PROBLEM STATEMENT}, we review some basic concepts and introduce the scenarios and threat models involved in this article. In  Section \ref{The S-SSI$_1$ Construction} and Section~\ref{The S-SSI$_2$ Construction}, we give the details of our  proposed constructions.  Next, performance evaluation  is presented in Section \ref{Performance Evaluation}. Finally, Section \ref{sec:conclusion} concludes the paper.

\section{BACKGROUND AND PRELIMINARY}
\label{sec:PROBLEM STATEMENT}
We first define the threat model considered in this paper, and then review some  tools and cryptographic primitives used in the proposed construction.
\subsection{Threat Model}
\label{secsec:Threat Model}
We consider a secure two-party computation model consisting of a client and a server. On the S-SIP computing platform, the client holds the dataset $(X, S)=\{(x_1, s_1), \cdots, (x_t, s_t)\}$, and the server holds the dataset $(Y, G)=\{(y_1, g_1), \cdots, (y_n, g_n)\}$. At the end of the calculation, the client and the server obtain the secret-sharing of  the inner product of the intersection of the two datasets, i.e., $f((X, S), (Y, G))=\sum_{i\in[t], j\in[n], x_i=y_j}s_ig_j$, where $[t]$ denotes the set $\{1, \cdots, t\}$. On the ML computing platform, we extend our  S-SIP to general ML scenarios to demonstrate its efficiency. It will be used as a basic component of linear algebra to execute linear operations including matrix-vector inner product and matrix multiplication. At the end of the computation, the client  gets the inference results and the server gets nothing. In the above scenarios both the client and the server are considered honest but curious, which is consistent with all previous work \cite{chen2021homomorphic,cui2021exploiting,schoppmann2019make,Xu2019TIFS}. Specifically, both parties follow the protocol's specifications but may infer the other's data privacy through passive acquisition of data flows during protocol execution. The security requirement of our S-SIP is to ensure that at the end of the protocol, the two parties only get the share of the inner product, and know nothing about their respective secret inputs.

\subsection{Secret Sharing and Oblivious Transfer}
\label{secsec:Secret Sharing and Oblivious Transfer}
\begin{packeditemize}
\item \textbf{Additive Secret Sharing}\cite{mishra2020delphi}. Without loss of generality, we assume that all variables involved in the paper lie in a prime field $\in \mathbb{F}_p$. Hence, given an arbitrary $x \in \mathbb{F}_p$, the  additive secret sharing of $x$ is denoted as a pair $(\left \langle x \right \rangle_0, \left \langle x \right \rangle_1)=(x-r, r)\in \mathbb{F}_{p}^{2}$, where $r$ is a random value uniformly selected from $\mathbb{F}_p$, and $x=\left \langle x \right \rangle_0+\left \langle x \right \rangle_1$.  Additive secret sharing is perfectly hiding, that is, given a share $\left \langle x \right \rangle_0$ or $\left \langle x \right \rangle_1$, $x$ is perfectly hidden.

\item \textbf{Oblivious Transfer (OT)} \cite{keller2018overdrive}. The 1-out-of-$n$ Oblivious Transfer (OT) is a two-party secure protocol, where the sender (defined as $P_0$) has $n$ inputs $(a_0, \cdots, a_n)$, while the receiver (defined as $P_1$) input a choice $b\in[n]$. At the end of the OT-execution, $P_1$ learns  $a_b$ while $P_0$ learns  nothing.
\end{packeditemize}

\subsection{Bloom Filter and Garbled Bloom Filter}
\label{secsec:Bloom Filter (BF) and Garbled Bloom Filter (GBF)}
\begin{packeditemize}
\item \textbf{Bloom Filter (BF)} \cite{dong2013private}. A Bloom filter is a compact data structure for probabilistic set membership testing. A BF is essentially a binary array of  $m$ bits that can be used to represent a set $Q$ of at most $n$ elements. Specifically, given a collection of $k$ hash functions $H=\{h_1, \cdots, h_k\}$, each of which maps an arbitrary element to the range $[1, m]$, i.e., $h_i: \{0, 1\}^{\ast}\rightarrow [m]$.   We denote the bit of BF at index $i$ by $BF[i]$. BF is first initialized with all bits in the array to 0. Then, to insert an element $x\in Q$ into BF, the element is first hashed through $k$ hash functions to obtain $k$ indices. All these indices corresponding to  BF will be assigned to  1, i.e.,  $BF[h_i(x)]=1$ for $1\leq i\leq k$.  Similarly, to verify whether an element $y\in Q$ is in the BF, $y$ is also hashed through $k$ hash functions, and then all locations $y$ hashes to are checked. If any of the bits at the locations is not 1, $y$ is not in set $Q$, otherwise $y$ probably  in $Q$. In this paper, we choose the optimal $k$ and $m$ so that once the above verification passes, then $y\in Q$ with overwhelming probability.

\item \textbf{Garbled Bloom Filter (GBF)} \cite{dong2013private}. GBF is similar in function to BF but it is an array of integers. It is used to store key-value pairs $(x, y)$, where $y$ is associated with key $x$ via $y=\sum_{i=1}^{k}GBF[h_i(x)]$. To achieve this, GBF and BF are first initialized with all entries as $\bot$ and $0$, respectively.  Then, for each key-values pair $(x, y)$, we  set $BF[h_j(x)]=1$ for all $j\in [k]$. Then, let $$B=\{h_j(x)|j\in[k], GBF[h_j(x)]=\bot\}$$ be the relevant positions of GBF that have not yet been set.  For $j\in B$,  we choose random values for $GBF[i]$, such that $\sum_{j=1}^{k}GBF[h_j(x)]=y$. For any remaining $GBF[j]=\bot$, $GBF[j]$ is set with random value uniformly chosen fron $\mathbb{F}_p$.
\end{packeditemize}

\subsection{Hashing Scheme}
\label{secsec:Hashing Scheme}
\begin{packeditemize}
\item \textbf{Cuckoo Hashing}\cite{menon2022spiral}. The basic Cuckoo hash contains $m$ bins, denoted as $C[1], \cdots, C[m]$. Given a stash, $k$ hash functions $h_i$, $i\in [k]$ which map any input to the range $[m]$, the workflow of inserting any element $x$ to the Cuckoo hash table is as follows: Calculate the $k$ candidate bins of $x$ by performing $k$ independent hashes for $x$. Place $x$ into an arbitrary empty candidate bin. If none of the $k$ candidate bins is empty, select one at random, remove the element  currently in that bin ($x_{old}$), place $x$ in the bin, and then reinsert the previously removed $x_{old}$. If re-inserting $x_{old}$ causes another element to be removed, this process continues recursively for a maximum number of iterations. %In this paper, we follow previous work to choose optimal parameters $m$ and $k$ such that each bin contains at most one element with high probability.

\item \textbf{2-choice hashing} \cite{ali2021communication}. 2-choice hashing is similar to Cuckoo hashing function except that instead of $k$ hash functions we only choose two hash functions $h_1$ and $h_2$. A 2-choice hashing algorithm assigns $x$ to whichever of $h_1(x)$, $h_2(x)$ has fewest elements.
\end{packeditemize}

\subsection{Fully Homomorphic Encryption}
\label{secsec:Fully Homomorphic Encryptioe}
Fully homomorphic encryption (FHE) \cite{choi2021compressed} enables the evaluation of arbitrary functions (parsed as polynomials) under ciphertext without decryption.
Let the plaintext space be $\mathbb{F}_p$,  a FHE under the public key encryption system usually contains the following algorithms:
\begin{packeditemize}
\item $\mathtt{KeyGen}(1^\lambda)\rightarrow (pk, sk)$.  Taking the security parameter $\lambda$ as input,  the algorithm $\mathtt{KeyGen}$ outputs the public-secret key pair ( $pk$, $sk$) required for fully homomorphic encryption.
\item $\mathtt{Enc}(pk, x)\rightarrow ct$.  Taking $pk$ and any  plaintext $x\in\mathbb{F}_p$, the algorithm $\mathtt{Enc}$ outputs  the ciphertext $ct$ of $x$.
\item $\mathtt{Dec}(sk, ct)\rightarrow x$.  Taking $sk$ and a ciphertext $c$ as input, the algorithm $\mathtt{Dec}$ outputs the decryption   corresponding plaintext $x$.
\item $\mathtt{Eval_{sum}}(pk, \{ct_i\})\rightarrow c'$. Taking $pk$ and a set of ciphertexts $\{ct_i=\mathtt{Enc}(pk, x_i)\}$ as input,  the algorithm $\mathtt{Eval_{sum}}$ outputs  a ciphertext $c'$  encrypting  $\sum_{i}x_i$.
\item $\mathtt{Eval_{mul}}(pk, \{ct_i\})\rightarrow c'$. Taking $pk$ and a set of ciphertexts $\{ct_i=\mathtt{Enc}(pk, x_i)\}$ as input,  the algorithm $\mathtt{Eval_{mul}}$ outputs  a ciphertext $c'$  encrypting  $\prod_{i}x_i$.
\end{packeditemize}
In this paper, we utilize homomorphic encryption methods BFV \cite{halevi2019improved} and BGV \cite{gentry2012ring} to implement the above homomorphic operations, which are constructed on the Ring Learning with Errors (RLWE) problem and have been well implemented by the mainstream libraries \cite{chen2017simple}.
\subsection{Private Information Retrieval}
\label{secsec:Private Information Retrieval}
In this paper, we leverage PIR technology on the single server as an underlying building block \cite{mughees2021onionpir}. Briefly, given a server holding a database $DB$ of $N$ strings, PIR enables the client to read an arbitrary entry $DB[i]$ without revealing $i$. Informally, PIR on a single server consists of the following algorithms:
\begin{packeditemize}
\item $\mathtt{PIR_{KeyGen}}(1^\lambda)\rightarrow (pk, sk)$.  Taking the security parameter $\lambda$ as input, the algorithm $\mathtt{KeyGen}$ outputs the public-secret key pair ( $pk$, $sk$) for the fully homomorphic encryption.
\item $\mathtt{PIR_{Query}}(pk, i)\rightarrow q$.  Taking $pk$ and a  plaintext index $i\in {N}$, the algorithm $\mathtt{PIR_{Query}}$ outputs  the ciphertext query $q$.
\item $\mathtt{PIR_{Answer}}(pk, q, DB)\rightarrow d$.  Taking $pk$, the  ciphertext query $q$ and $DB$ as input, the algorithm $\mathtt{PIR_{Answer}}$ outputs an answer $d$ encrypting the content of $DB[i]$.
\item $\mathtt{PIR_{Extract}}(sk, d)\rightarrow DB[i]$. Taking $sk$ and the answer $d$  as input,  the algorithm $\mathtt{PIR_{Extract}}$ outputs $DB[i]$.
\end{packeditemize}
\section{The S-SIP$_1$ Construction}
\label{The S-SSI$_1$ Construction}
In this section, we describe the first construct, S-SIP$_1$, which requires offline overhead linear to larger dataset sizes, while concomitant with superior online performance. The functionality of S-SIP$_1$ is depicted in Fig~\ref{Functionality S-SSI$_1$}, where the client holds the dataset $(X, S)=\{(x_1, s_1), \cdots, (x_t, s_t)\}$, and the server holds the dataset $(Y, G)=\{(y_1, g_1), \cdots, (y_n, g_n)\}$. At the end of the calculation, the client and the server obtains the shares of  the inner product of the intersection. We first give a high-level overview of our S-SIP$_1$, then we describe the technical details of S-SIP$_1$ and analyze its security.
\renewcommand\tablename{Fig}
\renewcommand \thetable{\arabic{table}}
\setcounter{table}{0}
\begin{table}
\begin{tabular}{|p{8cm}|}
\Xhline{1pt}
\begin{center}
$\mathtt{\pi}_{\rm S-SIP}$: Functionality of S-SIP
\end{center}
 \textbf{Input:}  The client (named $P_0$)  holds a set of $t$ pairs $(X, S)=\{(x_1, s_1), \cdots, (x_t, s_t)\}$, while the server (named $P_1$) holds dataset of key-values pairs $(Y, G)=\{(y_1, g_1), \cdots, (y_n, g_n)\}$\\
 \textbf{Output:} $P_b$ learns a set $\mathbf{U}_b=\{\left \langle \mathtt{u}_i\right \rangle_b\}_{i\in t}$, where $\left \langle \mathtt{u}_i\right \rangle_b=\left \langle s_ig_j\right \rangle_b$ if $x_i=y_j$ for some $j\in [n]$. otherwise $\left \langle \mathtt{u}_i\right \rangle_b=\left \langle 0\right \rangle_b$. \\
\\\Xhline{1pt}
\end{tabular}
\caption{Functionality of S-SIP}
\label{Functionality S-SSI$_1$}
\vspace{-20pt}
\end{table}

\subsection{Overview}
\label{sec:Overview}
As shown in Fig~\ref{Functionality S-SSI$_1$}, we assume that the server holds the dataset $(Y, G)$, and the client holds the dataset $(X, S)$. For each component $x_i\in X$, we aim at S-SIP$_1$ to compute  the secret share of $s_ig_j$ to both parties if $x_i = y_j$ for some $j$. Otherwise both hold  shares of $0$. The security requirements of S-SIP$_1$ require that at the end of protocol, the server has no knowledge of the client's inputs, and the client also knows nothing about  the dataset held by the server. The core  insight of S-SIP$_1$ lies in the fusion of BF and its variant GBF. Specifically, BF can be used to check the membership of  $x_i$ in the set represented by BF. It is implemented by accessing $k$ locations  $h_1(x_i), \cdots, h_k(x_i)$  in the BF  and checking that they are all 1 (or alternatively, checking $k=\sum_{j\in[k]}BF[h_j[x_i]]$).  GBF as a data structure similar to BF, it allows to store not only a set but also a set of associated values. Concretely, If $x_i$ is in the dataset held by the server, computing $\sum_{j\in[k]}GBF[h_j[x_i]]$ will result in the associated value (i.e., $g_i$). However, since $x_i$ is not in the dataset, $\sum_{j\in[k]}GBF[h_j[x_i]]$ implies a garbage value which will be converted to shares of 0.

Based on the properties of BF and GBF, S-SIP$_1$ can be divided into the following parts: in the \textbf{Offline Phase}, the server generates BF which is inserted all of  indexes of $Y$,      and a GBF which contains the database $(Y, G)$. In the online phase, the client interacts with the server with inputs  $\{(x_1, s_1), \cdots, (x_t, s_t)\}$. For each $x_i$, the two parties first run a \textbf{Secure Membership Check Protocol} to check whether $x_i$ is in the BF. The resulting membership bits will be held by both parties in the secret sharing way. Then, based on the previous results, the two parties perform the \textbf{Secure Associated Value Extraction Protocol}, which either gets secret sharing of the value $g_i$ associated with $x_i$ or  shares about 0. The client and the server further execute the \textbf{Secure Component Product Protocol}, where the outputs of the previous two protocols is used as input, and outputs the secret sharing of $s_ig_j$, if $x_i = y_j$ for some $j$, otherwise, outputs shares of 0.  Finally, the client and the server locally sum up all the obtained shares, and eventually obtains a secret-share of the inner product of two datasets' intersections, respectively.

\subsection{Technical Details of S-SIP$_1$}
\label{Technical Details of S-SSI$_1$}
\begin{table*}
\begin{tabular}{|p{17.5cm}|}
\Xhline{1pt}
\begin{center}
Implementation of S-SIP$_1$
\end{center}
 \textbf{Input:}  The client (named $P_0$)  holds a set of $t$  queries $X=\{ x_1, \cdots x_t\}$ associated with $t$ values $S=\{s_1, \cdots, s_t\}$.     The server (named $P_1$) holds dataset of key-values pairs $(Y, G)=\{(y_1, g_1), \cdots, (y_n, g_n)\}$\\
 \textbf{Implementation:}
 \begin{packeditemize}
 \item [1.]\textbf{Offline Phase:}
 \begin{packeditemize}
\item $P_0$ and $P_1$ negotiate $k$ hash functions $\{h_1, \cdots, h_k\}$ where $h_i: \{0, 1\}^{\ast} \rightarrow [m]$ and $m$ represents the size of the BF that is enough to insert $n$ entries.
\item Given the security parameter $\lambda$, $P_1$ invokes the algorithm $\mathtt{KeyGen}(1^\lambda)$  to generate key-pair $(pk, sk)$, and then sends $pk$ to $P_0$.
\item Using $k$ hash functions, $P_1$ inserts the set $Y$ containing keys $\{y_1, \cdots, y_n\}$ into BF, and also inserts set $(Y, G)=\{(y_1, g_1), \cdots, (y_n, g_n)\}$ containing key-value pairs into GBF. $P_1$ aborts if either insert operation fails.
\item Given $pk$, $P_1$ encrypts $BF$ and $GBF$ as $ct.BF[i]=\mathtt{Enc}(pk, BF[i])$ and $ct.GBF[i]=\mathtt{Enc}(pk, GBF[i])$, respectively, for every  $i \in [m]$.
\item $P_1$ sends $ct.BF$ and $ct.EBF$ to $P_0$.
\end{packeditemize}
\item [2.]\textbf{Online Phase:} $P_0$  interacts with $P_1$ to perform the following steps in parallel for every $x_j\in[t]$.
\begin{packeditemize}
\item [(a)] \textbf{Secure Membership Check Protocol}:
\begin{packeditemize}
\item  $P_0$ uniformly selects mask $\mu \in \mathbb{F}_p$ and computes $\nu=\sum_{i=1}^{k}ct.BF[h_i(x_j)]-\mathtt{Enc}(pk, \mu)$. Then, $P_0$ sends the ciphertext  $\nu$ to $P_1$.
\item $P_1$ obtains $\mu'$  by decrypting $\nu$ with the secret key $sk$.
\item $P_0$ and $P_1$ invoke an instance of 1-out-of-($k+1$) OT:
\begin{packeditemize}
\item [-] $P_1$  selects a random bit $b_{P_1}$. Then, $P_1$ as the OT's sender sets its inputs  to $\{b_0, \cdots, b_k\}$, where each $b_i$ is equal to $b_{P_1}$, except that $b_{(k-\mu')\mod (k+1)}$ is set equal to $1\bigoplus b_{P_1}$.
\item [-] $P_0$ as the OT's receiver inputs choice $\mu\mod (k+1)$, and then obtain $b_{P_0}$ from the OT's functionality.
\end{packeditemize}
\end{packeditemize}
\item[(b)]\textbf{Secure Associated Value Extraction Protocol}:
\begin{packeditemize}
\item  $P_0$ uniformly selects mask $\delta \in \mathbb{F}_p$ and computes $\nu'=s_j\cdot \sum_{i=1}^{k}ct.GBF[h_i(x_j)]-\mathtt{Enc}(pk, \delta)$. Then, $P_0$ sends the ciphertext  $\nu'$ to $P_1$.
\item $P_1$ obtains $\rho$  by decrypting $\nu'$ with the secret key $sk$.
\end{packeditemize}
\item [(c)] \textbf{Secure Component Product Protocol}:
\begin{packeditemize}
\item  $P_0$ and $P_1$ invoke an instance of 1-out-of-$2$ OT:
\begin{packeditemize}
\item [-] $P_0$  selects a random value $\Delta \in \mathbb{F}_p$. Then, $P_0$ acts as the OT's sender with two inputs $m_0=\Delta+b_{P_0}\cdot \delta$ and $m_1=\Delta+(1-b_{P_0})\cdot \delta$.
\item [-] $P_1$ as the OT's receiver inputs choice bit $b_{P_1}$, and then obtain $r$ from the OT's functionality. Note that $r=\Delta+b\cdot \delta$ where $b=b_{P_0}\bigoplus b_{P_1}$.
\end{packeditemize}
\item  $P_0$ and $P_1$ invoke another instance of 1-out-of-$2$ OT:
\begin{packeditemize}
\item [-] $P_1$ selects  random value $\alpha_j \in \mathbb{F}_p$, Then, $P_1$ acts as the OT's sender with two inputs $m_0=r+b_{P_1}\cdot (\rho-\alpha_j)-(1-b_{P_1})\alpha_j$ and $m_1=r+(1-b_{P_1})\cdot (\rho-\alpha_j)-(b_{P_1})\alpha_j$.
\item [-] $P_0$ as the OT's receiver inputs choice bit $b_{P_0}$, and then obtain $r'$ from the OT's functionality. Note that $r'=r+b\cdot (\rho-\alpha_j)+(1-b)\cdot\alpha_j$ where $b=b_{P_0}\bigoplus b_{P_1}$.
\end{packeditemize}
\item $P_0$ computes $r'-\Delta$, which implies that the output is exactly $\delta+\rho-\alpha_j$ if $b=1$. Otherwise, the output is $-\alpha_j$.
\item Since $P_1$ holds $\alpha_j$, the two parties holds the secret shares of $\delta+\rho$ if $b=1$, or shares of $0$ otherwise.
\end{packeditemize}
\end{packeditemize}
\end{packeditemize}\\
\Xhline{1pt}
\end{tabular}
\caption{Implementation of S-SIP$_1$}
\label{Implementation of S-SIP$_1$}
\vspace{-20pt}
\end{table*}
As described above, S-SIP$_1$ can be divided into offline phase and online phase, wherein the online phase contains three sub-protocols: \textbf{Secure Membership Check Protocol}, \textbf{Secure Associated Value Extraction Protocol} and \textbf{Secure Component Product Protocol}. Fig~\ref{Implementation of S-SIP$_1$} depicts the detailed technique for implementing S-SIP$_1$, and below we explain each step further.

\subsubsection{\quad\; Offline Phase}
\label{Offline Phase}
This phase requires the server to perform a series of offline operations that are independent of client input. This process is performed only once, and can be reused for multiple protocol executions, even for different clients. Specifically, the server first generates a public-secret key pair for homomorphic encryption, and $k$ hash functions for BF and GBF. The server then maps all entries in its own database into BF and GBF, using $k$ hash functions. At the end, the server performs homomorphic encryption on each entry in the BF and GBF and sends the result to the client (see step 1 in Fig~\ref{Implementation of S-SIP$_1$}).

\subsubsection{\quad\; Online Phase}
\label{Online Phase}
Given  queries $X=\{ x_1, \cdots, x_t\}$ associated with $t$ values $S=\{s_1, \cdots, s_t\}$, the client  interacts with server to perform the following steps in parallel for every $x_j\in[t]$.

(a) \textbf{Secure Membership Check Protocol}: For each $x_j\in[t]$, the client first computes  $\sum_{i=1}^{k}ct.BF[h_i(x_j)]$.
It is easy to observe that $\sum_{i=1}^{k}ct.BF[h_i(x_j)]$ is an encryption with a value less than $k+1$. Further, this value is equal to $k$ if $x_j$ is presented in the server's database $Y$. The purpose of this sub-protocol is to compute the membership bit of $x_j$ and share it secretly between two parties. To achieve this, a scarecrow approach is to use HE to convert membership $b$ into encryption for one bit ($0$ or $1$), which is then shared secretly to all parties. The entire conversion can be done by homomorphically evaluating an equality circuit, which has the multiplicative depth  $\lceil \log(k)\rceil$ resulting computationally expensive overhead.

Instead, we explore a simple approach based on oblivious transfer. Let $\sum_{i=1}^{k}ct.BF[h_i(x_j)]$  be the encryption of some plaintext $\eta$. The client sends $\nu=\sum_{i=1}^{k}ct.BF[h_i(x_j)]-\mathtt{Enc}(pk, \mu)$ to the server, where  the random value $\mu$ is treated as a secret-share of $\eta$ held by the client.  The server decrypts $\nu$ with secret key and obtains its share $\mu'=\eta-\mu$. Based on this, the client and server invoke an instance of 1-out-of-($k+1$) OT as below.

  The server first selects a random bit $b_{P_1}$. Then, it acts as the OT's sender and sets its inputs  to $\{b_0, \cdots, b_k\}$, where each $b_i$ is equal to $b_{P_1}$, except that $b_{(k-\mu')\mod (k+1)}$ is set to $1\bigoplus b_{P_1}$. On the other hand,  the client acts  as the OT's receiver and inputs choice $\mu\mod (k+1)$. At the end  of the OT execution, the functionality of OT ensures that the client gets the $b_{P_0}$, where $b_{P_0}\bigoplus b_{P_1}=1$ if $\mu+\mu'=k$, or $b_{P_0}\bigoplus b_{P_1}=0$, and the server gets nothing.  Therefore, the protocol described above achieve the  functionality  that the two parties obtain XOR shares of 1 or 0 if the client's query is or is not in the database.

(b) \textbf{Secure Associated Value Extraction Protocol}: This sub-protocol is used to compute the secret-shared associated value. It enables the client and server to hold a share of a value on the database, respectively, and this value corresponds to the client's current query. To achieve this, the client first computes  $\phi=s_j\cdot \sum_{i=1}^{k}ct.GBF[h_i(x_j)]$. Based on the property of GBF, $\sum_{i=1}^{k}ct.GBF[h_i(x_j)]$ is an encryption of associated value presented in server's database if $x_j=y_i$ for some $i\in[n]$. Then, the client uniformly selects mask $\delta \in \mathbb{F}_p$ and send   $\nu'=\phi-\mathtt{Enc}(pk, \delta)$  to the server, which decrypts $\nu'$ and obtains its share  $\rho$, where we can infer that $\delta+\rho=s_jg_i$ if $x_j=y_i$ for some $i\in[n]$. Note that if $x_j$ is not in the server's database $Y$, the above protocol gets a useless value which may be an arbitrary function of the server's database entries.

(c) \textbf{Secure Component Product Protocol}: This sub-protocol is used to compute the secret-shared component product. for each $x_j$ where $j\in [t]$, it enables the client and server to hold a shared share of $s_jg_i$, if $x_j=y_i$ for some $i\in[n]$, or shares of 0 otherwise. Specifically, both parties now  hold a share on $b$ and $\phi$ through the execution of the previous protocols. We translate the shares into the required output using 2 OT invocations:

 In the first OT, the client selects a random value $\Delta \in \mathbb{F}_p$. Then, it acts as the OT's sender with two inputs $m_0=\Delta+b_{P_0}\cdot \delta$ and $m_1=\Delta+(1-b_{P_0})\cdot \delta$. On the other hand,  the server as the OT's receiver inputs choice bit $b_{P_1}$, and then obtain $r$ from the OT's functionality. Clearly,  $r=\Delta+b\cdot \delta$ where $b=b_{P_0}\bigoplus b_{P_1}$. In the second $OT$, the server selects  random value $\alpha_j \in \mathbb{F}_p$, Then, it acts as the OT's sender with two inputs $m_0=r+b_{P_1}\cdot (\rho-\alpha_j)-(1-b_{P_1})\alpha_j$ and $m_1=r+(1-b_{P_1})\cdot (\rho-\alpha_j)-(b_{P_1})\alpha_j$. On the other hand, the client as the OT's receiver inputs choice bit $b_{P_0}$, and then obtain $r'$ from the OT's functionality. Clearly,  $r'=r+b\cdot (\rho-\alpha_j)+(1-b)\cdot\alpha_j$ where $b=b_{P_0}\bigoplus b_{P_1}$.

 Based on the two OTs, the client computes $r'-\Delta$, which implies that the output is exactly $\delta+\rho-\alpha_j$ if $b=1$. Otherwise, the output is $-\alpha_j$. Therefor, since the server holds $\alpha_j$, the two parties holds the secret shares of $\delta+\rho$ if $b=1$, or shares of $0$ otherwise.  To compute the shares of inner production, i.e., the share of $\sum_{i\in[t], j\in[n], x_i=y_j}s_ig_j$, it only requires two parties to sum up the resulting shares on a single component, respectively.

\subsubsection{\quad\; Security of S-SIP$_1$}
\label{Security of S-SIP$_1$}
Our S-SIP$_1$ is secure against the  honest-but-honest adversaries. We provide the following theorem.
\begin{theorem}
Let HE and OT used in the S-SIP$_1$  be secure against the  honest-but-honest adversaries. Then our  S-SIP$_1$ is secure against the  honest-but-honest client and server.
\end{theorem}

\begin{proof}
Let $\mathtt{\pi}_{\rm S-SIP}$ shown in Fig~\ref{Functionality S-SSI$_1$} be the functionality of of S-SIP$_1$. We demonstrate the security of S-SIP$_1$ against honest but curious adversaries with the simulation-based paradigm \cite{lindell2017simulate}.

\noindent\textbf{Semi-honest client security}.  We first analyze the case where adversary $\mathcal{A}$ compromises an honest but curious client. Specifically, we demonstrate the existence of such a polynomial-time simulation in $\mathtt{Sim_C}$, which is given access to the client's inputs and outputs. It simulates the client's view that is indistinguishable from the real view.

We show the indistinguishability between real and simulated views by the following hybrid arguments.
\begin{packeditemize}
\item  $\mathtt{Hyb_1}$: This corresponds to the real  protocol.
\item $\mathtt{Hyb_2}$: In this hybrid, instead of encrypting the original BF and GBF, the $\mathtt{Sim_C}$ randomly generates BF and GBF with the same length as the original, encrypts them with HE and sends them to the client. Because the client does not have the secret key $sk$ corresponding to the public key $pk$ of the HE, the semantic security of the HE guarantees the indistinguishability between this hybrid and the real view.
\item $\mathtt{Hyb_3}$: In this hybrid,  instead of following the real input, the $\mathtt{Sim_C}$  simulates the server  by randomly selecting a new bit $b_{P_1}'$ and setting all of the server's $k$ inputs in OT to $b_{P_1}'$. Since at the end  of the OT execution, client just gets a random bit at position $b_{P_0}$, the above modification just lets the client get another random bit that is indistinguishable from the original random bit. Therefore, the underlying cryptographic primitives of OT guarantee the indistinguishability of this hybrid from the real view.
\item $\mathtt{Hyb_4}$: In this hybrid,  instead of following the real input, the $\mathtt{Sim_C}$  simulates the server  by  setting  the server's  inputs in OT to  the outputs  corresponding the client in the $\mathtt{\pi}_{\rm S-SIP}$. This is possible because the $\mathtt{Sim_C}$  is allowed to access the output of the client in the ideal function. Due to the simulation-privacy of OT, this hybrid is indistinguishable from the real view.
\end{packeditemize}
\noindent\textbf{Semi-honest server security}. We now analyze the case where adversary $\mathcal{A}$ compromises an honest but curious server. Specifically, we demonstrate the existence of such a polynomial-time simulation in $\mathtt{Sim_S}$, which is given access to the server's inputs and outputs. It can simulate the server's view to make it indistinguishable from the real view.
\begin{packeditemize}
\item  $\mathtt{Hyb_1}$: This corresponds to the real  protocol.
\item $\mathtt{Hyb_2}$: In this hybrid, instead of computing $\sum_{i=1}^{k}ct.BF[h_i(x_j)]$ for each $x_j\in[t]$, the $\mathtt{Sim_S}$ encrypts randomly  strings and and sends it to server. Since in the real view, the ciphertext sent  to the server is homomorphically subtracted an random values uniformly chosen from $ \mathbb{F}_p$. Hence, the semantic security of the HE guarantees the indistinguishability between this hybrid and the real view.
 \item $\mathtt{Hyb_3}$:  In this hybrid, instead of computing $\sum_{i=1}^{k}ct.GBF[h_i(x_j)]$ for each $x_j\in[t]$, the $\mathtt{Sim_S}$ encrypts randomly  strings and and sends it to server.  Similarly, since in the real view, the ciphertext sent  to the server is homomorphically subtracted an random values uniformly chosen from $ \mathbb{F}_p$. Hence, the semantic security of the HE ensures the indistinguishability between this hybrid and  real view.
 \item $\mathtt{Hyb_4}$:  In this hybrid,  instead of following the real input, the $\mathtt{Sim_S}$  simulates the server  by  setting  the server's  inputs in OT  with two random strings. This stems from the fact that in the real view, the input to the server is two statistically uniform random strings, and the security of OT guarantees that the client receives one of the two strings and knows nothing about the other. As a result, we hold the same security by  substituting  the original input with two new random strings. Therefore, this hybrid  is indistinguishable from the real view.
\end{packeditemize}
\end{proof}

\section{The S-SIP$_2$ Construction}
\label{The S-SSI$_2$ Construction}
We now describe our second construction S-SIP$_2$, which is a fully online setting  without precomputing. We instantiate S-SIP$_2$ with state-of-the-art PIR technology as the underlying technology. As a result, this derives the client's overhead asymptotically linear to its own input and logarithmic to the size of the server's database. Thus, it shifts the vast majority of the protocol overhead from the client to the server, which is beneficial in real-world applications where the client is usually a resource-constrained device such as a mobile phone.

\subsection{Sum-PIR Functionality}
\label{sec:Sum-PIR Functionality}
S-SIP$_2$ is functionally identical to S-SIP$_1$, but removes the expensive offline phase of S-SIP$_1$, replacing it with standard private information retrieval queries. Recall that during the offline phase of S-SIP$_1$, the server is required to encrypt the BF and GBF containing all database entries, i.e.,  $ct.BF[i]=\mathtt{Enc}(pk, BF[i])$ and $ct.GBF[i]=\mathtt{Enc}(pk, GBF[i])$ for every  $i \in [m]$,  and send them to the client. For each query $x_j$ by the client, the client is required to homomorphically sum all entries corresponding to position $h_i(x_j)$, i.e., $\sum_{i=1}^{k}ct.BF[h_i(x_j)]$ and $\sum_{i=1}^{k}ct.GBF[h_i(x_j)]$,  then masks the results and sends them to the server. In S-SIP$_2$, we instead use PIR to obliviously query the server for entries located at $h_i(x_j)$, and receive the masked sum of the corresponding values at those locations in BF and GBF. If the client only needs to retrieve the entry at $h_i(x_j)$  without summing and masking, it only needs to utilize the standard symmetric PIR. Whereas in S-SIP$_2$  the client needs to sum the values of the $k$ positions that the hashes map to, we use a modified version of PIR \cite{mughees2021onionpir}, named Sum-PIR.

\begin{table}
\begin{tabular}{|p{8cm}|}
\Xhline{1pt}
\begin{center}
Sum-PIR
\end{center}
 \textbf{Input:}  The client   holds a set of $k$ indices $\{\zeta_1, \cdots \zeta_k\}$, while the server  holds dataset $DB$ of  size $m$.\\
 \textbf{Protocol:}\\
 1. The client generates the key pair $(pk, sk)$ with $\mathtt{PIR_{KeyGen}}$, and then sends $pk$ to the server. \\
 2.  For each $i\in[k]$, the client interacts with the server with multi-query PIR as follows:
\begin{packeditemize}
\item [a.]  The clients generates a query $q_i$ by $\mathtt{PIR_{Query}}(pk, \zeta_i)$, and sends it to the server.
\item [b.] The server generates the answer $d_i$ with $\mathtt{PIR_{Answer}}(pk, q_i, DB)$, and then computes $ct=\sum_{i=1}^{k}d_i$.
\item [c.] The server homomorphically computes $ct^{\ast}\leftarrow ct-\mathtt{Enc}(pk, r)$, where $r$ is a value chosen at random.
\item [d.]  The server sends $ct^{\ast}$ to the client.
\end{packeditemize}
3. The client executes $\mathtt{PIR_{Extract}}(sk, ct^{\ast})$ to obtain $v=\sum_{i=1}^{k}D[\zeta_i]-r$.
\\
\\\Xhline{1pt}
\end{tabular}
\caption{Construction of Sum-PIR}
\label{Construction of Sum-PIR}
\vspace{-20pt}
\end{table}

Fig~\ref{Construction of Sum-PIR} depicts the construction of Sum-PIR, which allows a client holding $k$ indices to interact with the server to obtain $\sum_{i=1}^{k}D[\zeta_i]-r$, where the $r$  is an additive  mask randomly chosen by the server. At the end of the protocol execution, the server has no knowledge of the indexes held by the client.
\begin{table*}[htb]
\begin{tabular}{|p{17.5cm}|}
\Xhline{1pt}
\begin{center}
Implementation of S-SIP$_2$
\end{center}
 \textbf{Input:}  The client (named $P_0$)  holds a set of $t$  queries $X=\{ x_1, \cdots x_t\}$ associated with $t$ values $S=\{s_1, \cdots, s_t\}$.     The server (named $P_1$) holds dataset of key-values pairs $(Y, G)=\{(y_1, g_1), \cdots, (y_n, g_n)\}$\\
 \textbf{Implementation:}
 \begin{packeditemize}
 \item [1.]\textbf{Setup Phase:}
 \begin{packeditemize}
\item $P_0$ and $P_1$ negotiate $k$ hash functions $\{h_1, \cdots, h_k\}$ where $h_i: \{0, 1\}^{\ast} \rightarrow [m]$ and $m$ represents the size of the BF that is enough to insert $n$ entries.
\item Using $k$ hash functions, $P_1$ inserts the set $Y$ containing keys $\{y_1, \cdots, y_n\}$ into BF, and also inserts set $(Y, G)=\{(y_1, g_1), \cdots, (y_n, g_n)\}$ containing key-value pairs into GBF. $P_1$ aborts if either insert operation fails.
\end{packeditemize}
\item [2.]\textbf{Online Phase:} $P_0$  interacts with $P_1$ to perform the following steps in parallel for every $x_j\in[t]$.
\begin{packeditemize}
\item [(a)] \textbf{Secure Membership Check Protocol}:
\begin{packeditemize}
\item  $P_1$ uniformly selects mask $\mu \in \mathbb{F}_p$ and then $P_1$ interacts with $P_0$ to execute a Sum-PIR query, where the inputs of $P_0$ are $h_1(x_j), \cdots, h_k(x_j)$ while $P_1$ use BF and $\mu $ as input.
\item $P_0$ obtains $\mu'=-\mu+\sum_{i=1}^{k}BF[h_i(x_j)]$ as output.
\item $P_0$ and $P_1$ invoke an instance of 1-out-of-($k+1$) OT:
\begin{packeditemize}
\item [-] $P_1$  selects a random bit $b_{P_1}$. Then, $P_1$ as the OT's sender sets its inputs  to $\{b_0, \cdots, b_k\}$, where each $b_i$ is equal to $b_{P_1}$, except that $b_{(k-\mu)\mod (k+1)}$ is set equal to $1\bigoplus b_{P_1}$.
\item [-] $P_0$ as the OT's receiver inputs choice $\mu'\mod (k+1)$, and then obtain $b_{P_0}$ from the OT's functionality.
\end{packeditemize}
\end{packeditemize}
\item[(b)]\textbf{Secure Associated Value Extraction Protocol}:
\begin{packeditemize}
\item  $P_0$  computes $\mathtt{Enc}(pk, s_j)$ and sends it to the server.
\item $P_1$ uniformly selects mask $\rho \in \mathbb{F}_p$ and then $P_1$ interacts with $P_0$ to execute a Sum-PIR query, where the inputs of $P_0$ are $h_1(x_j), \cdots, h_k(x_j)$ while $P_1$ use GBF and $\rho$ as input.
\item  Before adding the additive  mask  $\rho$ to the result $ct$ obtained by executing Sum-PIR, the server  homomorphically multiplies $ct$ with $\mathtt{Enc}(pk, s_j)$.
\item The server takes $\rho$ as output and the client receives $\delta=-\rho+s_j\cdot \sum_{i=1}^{k}GBF[h_i(x_j)]$ as output.
\end{packeditemize}
\item [(c)] \textbf{Secure Component Product Protocol}:
\begin{packeditemize}
\item  $P_0$ and $P_1$ invoke an instance of 1-out-of-$2$ OT:
\begin{packeditemize}
\item [-] $P_0$  selects a random value $\Delta \in \mathbb{F}_p$. Then, $P_0$ acts as the OT's sender with two inputs $m_0=\Delta+b_{P_0}\cdot \delta$ and $m_1=\Delta+(1-b_{P_0})\cdot \delta$.
\item [-] $P_1$ as the OT's receiver inputs choice bit $b_{P_1}$, and then obtain $r$ from the OT's functionality. Note that $r=\Delta+b\cdot \delta$ where $b=b_{P_0}\bigoplus b_{P_1}$.
\end{packeditemize}
\item  $P_0$ and $P_1$ invoke another instance of 1-out-of-$2$ OT:
\begin{packeditemize}
\item [-] $P_1$ selects  random value $\alpha_j \in \mathbb{F}_p$, Then, $P_1$ acts as the OT's sender with two inputs $m_0=r+b_{P_1}\cdot (\rho-\alpha_j)-(1-b_{P_1})\alpha_j$ and $m_1=r+(1-b_{P_1})\cdot (\rho-\alpha_j)-(b_{P_1})\alpha_j$.
\item [-] $P_0$ as the OT's receiver inputs choice bit $b_{P_0}$, and then obtain $r'$ from the OT's functionality. Note that $r'=r+b\cdot (\rho-\alpha_j)+(1-b)\cdot\alpha_j$ where $b=b_{P_0}\bigoplus b_{P_1}$.
\end{packeditemize}
\item $P_0$ computes $r'-\Delta$, which implies that the output is exactly $\delta+\rho-\alpha_j$ if $b=1$. Otherwise, the output is $-\alpha_j$.
\item Since $P_1$ holds $\alpha_j$, the two parties holds the secret shares of $\delta+\rho$ if $b=1$, or shares of $0$ otherwise.
\end{packeditemize}
\end{packeditemize}
\end{packeditemize}\\
\Xhline{1pt}
\end{tabular}
\caption{Implementation of S-SIP$_2$}
\label{Implementation of S-SIP$_2$}
\vspace{-20pt}
\end{table*}
\subsection{Technical Details of S-SIP$_2$}
\label{Technical Details of S-SSI$_2$}
 With the properties of Sum-PIR, we now describe the technical details of S-SIP$_2$.   Similar to S-SIP$_1$, S-SIP$_2$ can be divided into setup phase and online phase, wherein the online phase also contains three sub-protocols: \textbf{Secure Membership Check Protocol}, \textbf{Secure Associated Value Extraction Protocol} and \textbf{Secure Component Product Protocol}. Fig~\ref{Implementation of S-SIP$_2$} depicts the detailed technique for implementing S-SIP$_2$, and below we explain each step further.

 \subsubsection{\quad\; Setup Phase}
\label{Setup Phase}
This process only requires the server to initialize its own dataset. It is functionally different from precomputing in S-SIP$_1$,  since it does not require the server to encrypt the local database and send it to the client. Specifically, the server first generates $k$ hash functions $\{h_1, \cdots, h_k\}$,  and then inserts the set $Y$ containing keys $\{y_1, \cdots, y_n\}$ into BF with these hash functions.  The set $(Y, G)=\{(y_1, g_1), \cdots, (y_n, g_n)\}$  is also inserted  into GBF with the similar way.

\subsubsection{\quad\; Online Phase}
\label{Online Phase2}
The client interacts with the server to perform the following steps in parallel for every $x_j\in[t]$.

(a) \textbf{Secure Membership Check Protocol}: For each $x_j\in[t]$, the server first uniformly selects mask $\mu \in \mathbb{F}_p$. Then, the client interacts with  server to execute a Sum-PIR query, where the inputs of client are $h_1(x_j), \cdots, h_k(x_j)$ while the server use BF and $\mu $ as input. As a result, the client obtains  $\mu'=-\mu+\sum_{i=1}^{k}BF[h_i(x_j)]$ as output. Afterwards, the client and server invoke an instance of 1-out-of-($k+1$) OT as below: the server first  selects a random bit $b_{P_1}$. Then, the server as the OT's sender sets its inputs  to $\{b_0, \cdots, b_k\}$, where each $b_i$ is equal to $b_{P_1}$, except that $b_{(k-\mu)\mod (k+1)}$ is set equal to $1\bigoplus b_{P_1}$. The client as the OT's receiver inputs choice $\mu'\mod (k+1)$. At the end  of the OT execution, the functionality of OT ensures that the client gets the $b_{P_0}$, where $b_{P_0}\bigoplus b_{P_1}=1$ if $\mu+\mu'=k$, or $b_{P_0}\bigoplus b_{P_1}=0$, and the server gets nothing.

(b) \textbf{Secure Associated Value Extraction Protocol}: This sub-protocol is used to compute the secret-shared associated value. To achieve this,   the client first  computes $\mathtt{Enc}(pk, s_j)$ and sends it to the server. Then.  the server uniformly selects mask $\rho \in \mathbb{F}_p$ and interacts with the client to execute a Sum-PIR query, where the inputs of the client  are $h_1(x_j), \cdots, h_k(x_j)$ while the server use GBF and $\rho$ as input. As a result,  before adding an  additive  mask  $\rho$ to the result $ct$ obtained by executing Sum-PIR, the server  homomorphically multiplies $ct$ with $\mathtt{Enc}(pk, s_j)$. Finally,  the client receives $\delta=-\rho+s_j\cdot \sum_{i=1}^{k}GBF[h_i(x_j)]$ as output.

(c) \textbf{Secure Component Product Protocol}: This sub-protocol is used to compute the secret-shared component product. for each $x_j$ where $j\in [t]$, it enables the client and server to hold a shared share of $s_jg_i$, if $x_j=y_i$ for some $i\in[n]$, or shares of 0 otherwise.  Its workflow is exactly the same as the corresponding steps in S-SIP$_1$. We omit here to prevent redundancy.

\textit{Remark}: The efficiency of S-SIP$_2$ relies  heavily on the performance of the PIR query. Since  for each $x_j\in[t]$ in S-SIP$_2$, we need to perform $k$  queries to obtain the sum of the ciphertext at the corresponding position. This is very time consuming as the number of $t$ increases. To get rid of this dilemma, we design an optimized PIR for batch queries to speed up execution (shown in Fig~\ref{Implementation of  Optimized S-SIP$_2$}). The main idea of this comes from parting the server's database into $m$ bins. The specific partition operation can be done using Cuckoo hashing or 2-choice hashing. In this way, each bin contains only a small part of the database, which allows parties to evaluate S-SIP$_2$ bin-by-bin. The amount of data the server has to touch with each query is now just the entries mapped into the same bin as the client query, which is computationally more efficient. Variants of this idea have  been  used in previous work SealPIR \cite{angel2018pir}.
\begin{table}[htb]
\begin{tabular}{|p{7.5cm}|}
\Xhline{1pt}
\begin{center}
Implementation of the optimized S-SIP$_2$
\end{center}
\textbf{Parameters}:
\begin{packeditemize}
\item The server's dataset with size of $n$, associated values space $\mathbb{F}_p$, the number of queries $t$.
\item The S-SIP$_2$ primitive.
\item The maximum number of bins is $m$, where the maximum size of each bin in the server is $\beta$ while the size of the client is $\eta$.
\item The number of hash functions $k$
\end{packeditemize}
 \textbf{Input:}  The client (named $P_0$)  holds a set of $t$  queries $X=\{ x_1, \cdots x_t\}$ associated with $t$ values $S=\{s_1, \cdots, s_t\}$.     The server (named $P_1$) holds dataset of key-values pairs $\mathcal{G}=\{(y_1, g_1), \cdots, (y_n, g_n)\}$\\
 \textbf{Implementation:}
 \begin{packeditemize}
 \item [1.] $P_0$ partitions it items $\{ x_1, \cdots x_t\}$ into $m$ bins with the Cuckoo or 2-choice hashing scheme. Without loss of generality, we denote by $B_C[b]$ those items in the $b$-th bin of the client.

\item [2.] $P_1$ partitions it items $\{ y_1, \cdots y_n\}$ into $m$ bins with the $k$ hash functions. Without loss of generality, we denote by $B_S[b]$ those items in the $b$-th bin of the client.
 \item for each bin $b\in[m]$
\begin{packeditemize}
\item [(a)] $P_1$ computes $\mathcal{G}_b=\{(y_i, l_i)|(y_i, l_i)\in \mathcal{G}$, $y_i\in B_S[b]\}$. Besides, $P_1$ pads $\mathcal{G}_b$ to the maximum bin size $\beta$ with dummy pairs.
\item [(b)]$P_0$ and $P_1$  invoke an instance of S-SIP$_2$, where the inputs of each party is as follows:
\begin{packeditemize}
\item  $P_0$ takes a set of $\eta$ queries $\{v_i|v_i\in B_C[b]\}$ which is padded with dummy items to the size $\eta$. Similarly, there is a set of associated value $\{s_i|v_i\in B_C[b], s_i\in S\}$, which is also padded with dummy items to the size $\eta$.
\item $P_1$ takes the set $\mathcal{G}_b$ as the inputs.
\end{packeditemize}
\item[(c)] $P_0$ receives the S-SIP$_2$'s outputs.
\end{packeditemize}
\end{packeditemize}\\
\Xhline{1pt}
\end{tabular}
\caption{Implementation of optimized S-SIP$_2$}
\label{Implementation of  Optimized S-SIP$_2$}
\vspace{-20pt}
\end{table}

\subsubsection{\quad\; Security of S-SIP$_2$}
\label{Security of S-SIP$_2$}
Our S-SIP$_2$ is secure against the  honest-but-honest adversaries. We provide the following theorem.
\begin{theorem}
Let PIR and OT used in the S-SIP$_2$  be secure against the  honest-but-honest adversaries. Then our  S-SIP$_2$ is secure against the  honest-but-honest client and server.
\end{theorem}

\begin{proof}
Let $\mathtt{\pi}_{\rm S-SIP}$ shown in Fig~\ref{Functionality S-SSI$_1$} be the functionality of S-SIP$_2$. We demonstrate the security of S-SIP$_2$ against honest but curious adversaries with the simulation-based paradigm.

\noindent\textbf{Semi-honest client security}.  We first analyze the case where adversary $\mathcal{A}$ compromises an honest but curious client. Specifically, we demonstrate the existence of such a polynomial-time simulation in $\mathtt{Sim_C}$, which is given access to the client's inputs and outputs. It can simulate the client's view to make it indistinguishable from the real view.

We show the indistinguishability between real and simulated views by the following hybrid arguments.
\begin{packeditemize}
\item  $\mathtt{Hyb_1}$: This corresponds to the real  the protocol.
\item $\mathtt{Hyb_2}$: In this hybrid, instead of executing Sum-PIR with real BF held on the server, the $\mathtt{Sim_C}$ randomly generates BF  and  $\mu \in \mathbb{F}_p$  with the same length as the original. $\mathtt{Sim_C}$ then interacts with the client to perform Sum-PIR queries. Since the output of Sum-PIR is masked before sending to the client, the security of Sum-PIR guarantees  the indistinguishability between this result and the one actually obtained.

\item $\mathtt{Hyb_3}$: In this hybrid,  instead of following the real input, the $\mathtt{Sim_C}$  simulates the server  by randomly selecting a new bit $b_{P_1}'$ and setting all of the server's $k$ inputs in OT to $b_{P_1}'$. Since at the end  of the OT execution, client just gets a random bit at position $b_{P_0}$, the above modification just lets the client get another random bit that is indistinguishable from the original random bit. Therefore, the underlying cryptographic primitives of OT guarantee the indistinguishability of this hybrid from the real view.
\item $\mathtt{Hyb_4}$: In this hybrid, instead of executing Sum-PIR with real GBF held on the server, the $\mathtt{Sim_C}$ randomly generates GBF  and  $\rho \in \mathbb{F}_p$  with the same length as the original. $\mathtt{Sim_C}$ then interacts with the client to perform Sum-PIR queries. Since the output of Sum-PIR is masked before sending to the client, the security of Sum-PIR guarantees  the indistinguishability between this result and the one actually obtained.
\item $\mathtt{Hyb_5}$: In this hybrid,  instead of following the real input, the $\mathtt{Sim_C}$  simulates the server  by  setting  the server's  inputs in OT to  the outputs  corresponding the client in the $\mathtt{\pi}_{\rm S-SIP}$. This is possible because the $\mathtt{Sim_C}$  is allowed to access the output of the client in the ideal function. Due to the simulation-privacy of OT, this hybrid is indistinguishable from the real view.
\end{packeditemize}
\noindent\textbf{Semi-honest server security}. We now analyze the case where adversary $\mathcal{A}$ compromises an honest but curious server. Specifically, we demonstrate the existence of such a polynomial-time simulation in $\mathtt{Sim_S}$, which is given access to the server's inputs and outputs. It can simulate the server's view to make it indistinguishable from the real view.
\begin{packeditemize}
\item  $\mathtt{Hyb_1}$: This corresponds to the real  protocol.
\item $\mathtt{Hyb_2}$: In this hybrid, instead of executing Sum-PIR with real input held on the client, the $\mathtt{Sim_S}$ encrypts $0$s  with the same length as the original. $\mathtt{Sim_S}$ then sends it to the server  to perform Sum-PIR queries. Since the server does not have the secret key $sk$ corresponding to the public key $pk$ used in $\mathtt{Sim_S}$, the semantic security of the Sum-PIR ensures the indistinguishability between this hybrid and real view.
 \item $\mathtt{Hyb_3}$:  In this hybrid, instead of executing Sum-PIR with real $s_j$ held on the client, the $\mathtt{Sim_S}$ randomly generates a new random $s_j'$   with the same length as the original. The security of Sum-PIR guarantees  the indistinguishability between this result and the one actually obtained.
 \item $\mathtt{Hyb_4}$:  In this hybrid,  instead of following the real input, the $\mathtt{Sim_S}$  simulates the server  by  setting  the server's  inputs in OT  with two random strings. This stems from the fact that in the real view, the input to the server is two statistically uniform random strings, and the security of OT guarantees that the client receives one of the two strings and knows nothing about the other. As a result, we hold the same security by  substituting  the original input with two new random strings. Therefore, this hybrid  is indistinguishable from the real view.
\end{packeditemize}
\end{proof}

\section{Performance Evaluation}
\label{Performance Evaluation}
In this section we discuss the performance of the two proposed constructs, S-SIP$_1$ and S-SIP$_2$. We use the work ROOM  as the baseline for comparison, as it is consistent with our motivation to design general-purpose secure sparse linear algebra. Below we first analyze the overhead of our schemes and ROOM \cite{schoppmann2019make} \footnote{Codes are available at https://github.com/schoppmp/room-framework} for performing sparse inner products on different sizes of dataset, and then compare the overhead of the two for performing different machine learning tasks including K-nearest neighbors, logistic regression, and naive Bayes classification.

\subsection{Implementation Details}
\label{Implementation Details}
We use SEAL \cite{sealcrypto} to implement homomorphic encryption for BF and GBF, where the polynomial dimension on the ring is set to 2048 and the ciphertext space parameter is $2^{46}$. It provides 128-bit security. We adopt OnionPIR \cite{mughees2021onionpir} as the underlying structure for constructing S-SIP$_2$. The realization of ROOM follows all the implemtation described in their paper. It uses Obliv-C \cite{zahur2015obliv} to implement the garbled circuit to construct a general secure two-party protocol, and Pseudo-random functions are constructed through the implementation of AES-128 \cite{schoppmann2019make}. Our experiments are carried out in both the LAN and WAN settings. LAN is implemented with two workstations in our lab. The client workstation has AMD EPYC 7282 1.4GHz CPUs with single core and 8GB RAM. The server workstation has Intel(R) Xeon(R) E5-2697 v3 2.6GHz CPUs with 28 threads on 14 cores and 64GB RAM. The WAN setting is based on a connection between a local PC and an Amazon AWS server with an average bandwidth of 963Mbps and running time of around 35ms.

\renewcommand\tablename{TABLE}
\renewcommand \thetable{\Roman{table}}
\setcounter{table}{0}
\setcounter{figure}{5}
\begin{table}[htb]
\centering
\caption{Real-world datasets used in the experiments}
\label{Real-world datasets used in the experiments}\
\setlength{\tabcolsep}{0.4mm}{
\begin{tabular}{|c|c|c|c|c|}
\Xhline{1pt}
{Dataset}&Documents&Classes &Nonzero Features(aveg.)&Total Features\\
\Xhline{1pt}
Movies&14341&2&136&95626\\
\hline
Newsgroups&9051&20&98&101631\\
\hline
Languages, ngrams=1&783&11&43&1033\\ \hline
Languages, ngrams=2&783&11&231&9915\\
\Xhline{1pt}
\end{tabular}}
\vspace{-10pt}
\end{table}
Consistent with ROOM, we chose three typical datasets (i.e. Movies \cite{maas2011learning}, Newsgroups \cite{albishre2015effective}, Languages \cite{Scikitlearn} with ngrams$=1$ and $2$, respectively) to implement ML tasks including $k$-nearest neighbors, logistic regression, and naive bayes classification in a privacy-preserving manner. Please refer to TABLE~\ref{Real-world datasets used in the experiments} for the specific size and sparsity of the dataset, and see ROOM for more details on the usage of these dataset.
\subsection{Performance of Executing Sparse Inner Products}
\label{Performance of Executing Sparse Inner Products}

\begin{table*}[!htbp]
%\footnotesize
\centering
\caption{Cost of SIP with different size of datasets}
\label{Cost of SIP with Different Size of Datasets}
\begin{tabular}{|cc|c|c|c|c|c|c|c|c|}
\Xhline{1pt}
\multicolumn{2}{|c|}{\multirow{2}*{Parameters}}&\multicolumn{4}{c|}{S-SIP$_1$}&\multicolumn{2}{c|}{S-SIP$_2$}& \multicolumn{2}{c|}{ROOM}\\\cline{3-10}
&&\multicolumn{2}{c|}{Offline}&\multicolumn{2}{c|}{Online}&\multicolumn{2}{c|}{Online}& \multicolumn{2}{c|}{Online}\\\cline{1-10}
$n$&$t$&Comm.(MB)&Time (s)&Comm.(MB)&Time&Comm.(MB)&Time&Comm.(MB)&Time (s)\\\cline{1-10}
\multirow{3}*{$2^{16}$}&$2^{8}$&$29$&$9.18$&$7$&$0.62$&$27$&$12.1$&$56$&$21.2$ \\ \cline{3-10}
&$2^{12}$&$29$&$9.18$&$112$&$4.22$&$120$&$84.23$&$863$&$339.9$ \\ \cline{3-10}
&$2^{16}$&$29$&$9.18$&$1794$&$27.5$&$801$&$623.2$&$13788$&$5439.4$ \\ \cline{1-10}
\multirow{3}*{$2^{18}$}&$2^{8}$&$116$&$36.8$&$7$&$0.62$&$29$&$16.4$&$71$&$31.2$ \\ \cline{3-10}
&$2^{12}$&$116$&$36.8$&$112$&$4.22$&$213$&$132.8$&$878$&$495.7$ \\ \cline{3-10}
&$2^{16}$&$116$&$36.8$&$1794$&$27.5$&$1821$&$1008$&$13837$&$7929.8$ \\ \cline{1-10}
\multirow{3}*{$2^{20}$}&$2^{8}$&$465$&$146.9$&$7$&$0.62$&$44$&$31.74$&$91$&$45.3$ \\ \cline{3-10}
&$2^{12}$&$465$&$146.9$&$112$&$4.22$&$379$&$231.7$&$1401$&$724.9$ \\ \cline{3-10}
&$2^{16}$&$465$&$146.9$&$1794$&$27.5$&$3704$&$1691.6$&$14391$&$11598.7$ \\ \Xhline{1pt}
\end{tabular}
\vspace{-15pt}
\end{table*}

\begin{figure*}[htb]
  \centering
  \subfigure[]{\label{6a}\label{Reprint_accuracy}\includegraphics[width=0.49\textwidth]{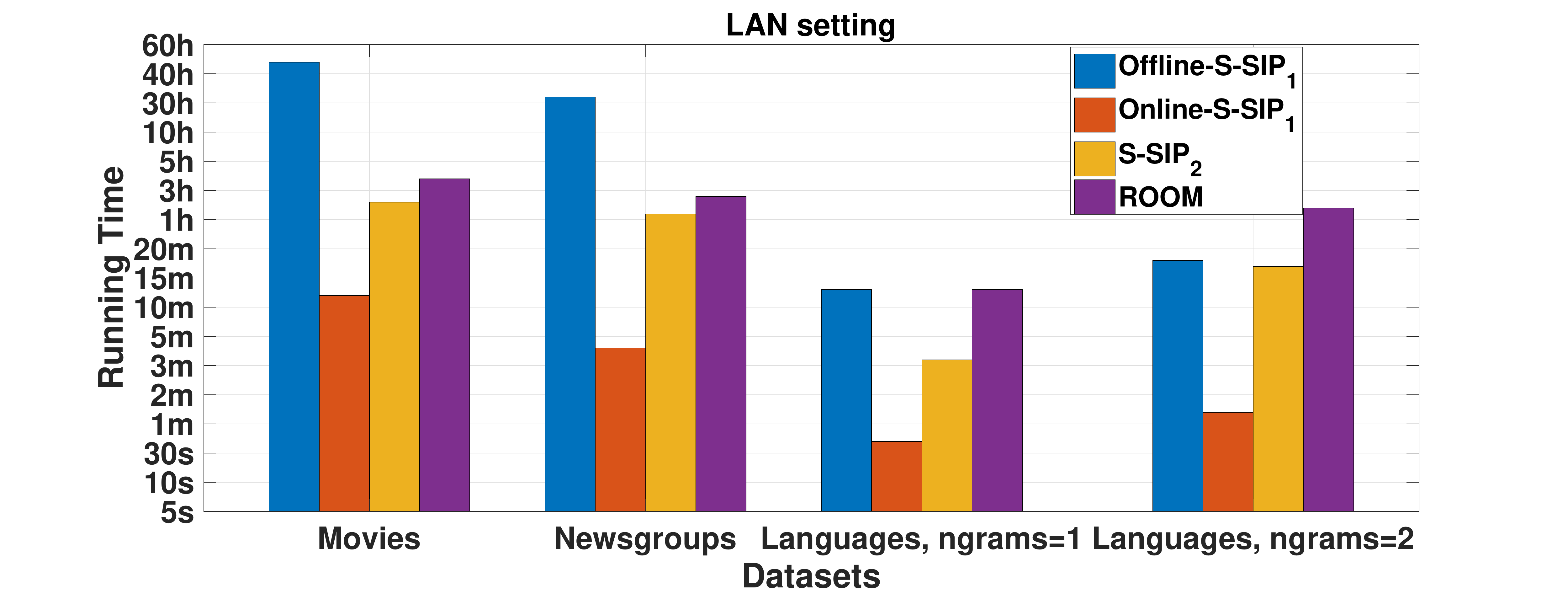}}
 \subfigure[]{\label{6b}\includegraphics[width=0.49\textwidth]{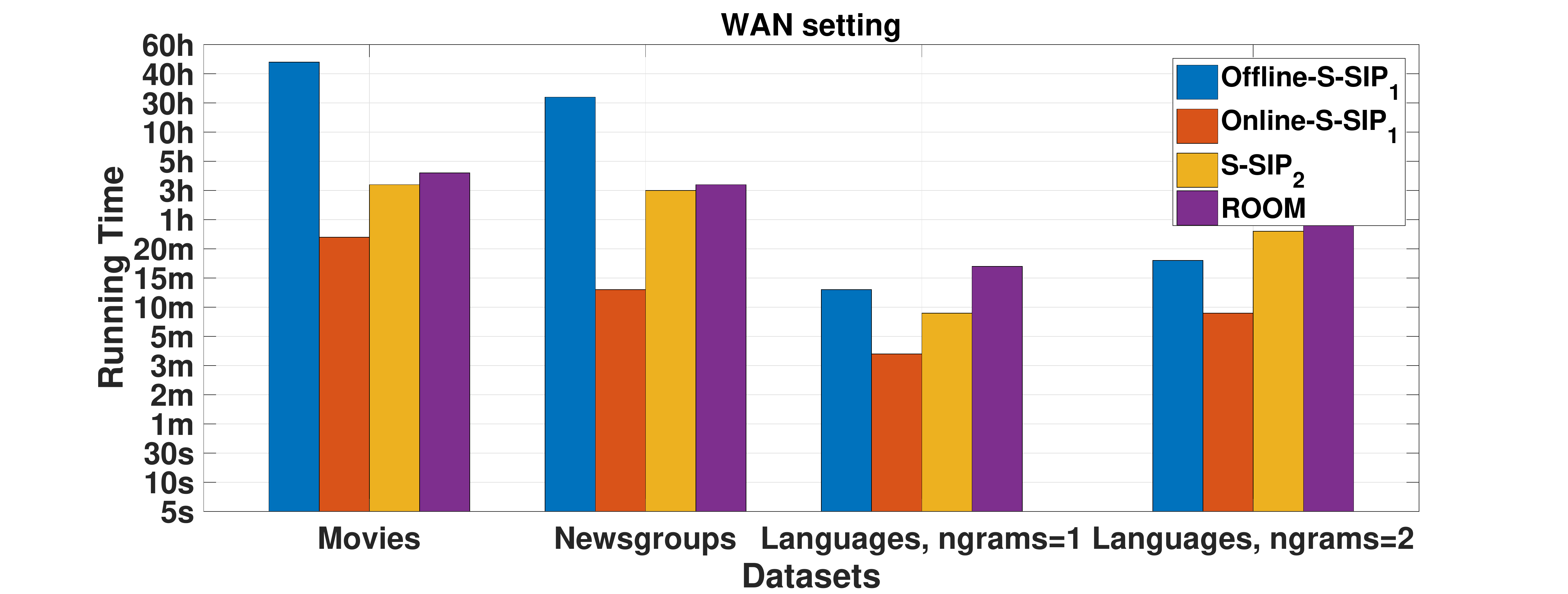}}
  \caption{Running time of kNN on different datasets. (a) LAN setting.  (b) WAN setiing. }
  \label{runtime on KNN}
  \vspace{-20pt}
\end{figure*}

\begin{figure}[htb]
\centering
\includegraphics[width=0.5\textwidth]{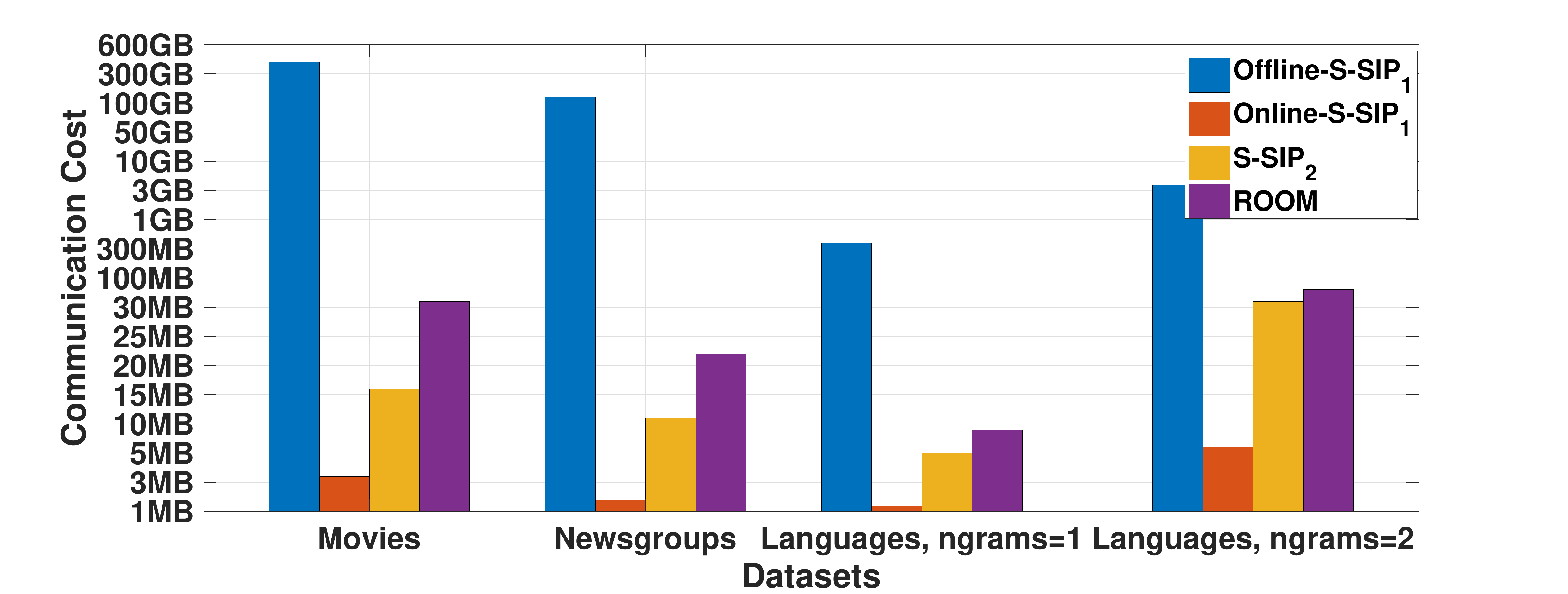}
\caption{Communication cost of kNN  on different datasets}
\label{Fig:KNN_comm}
 \vspace{-20pt}
\end{figure}

\begin{figure*}
  \centering
  \subfigure[]{\label{6a}\label{Reprint_accuracy}\includegraphics[width=0.49\textwidth]{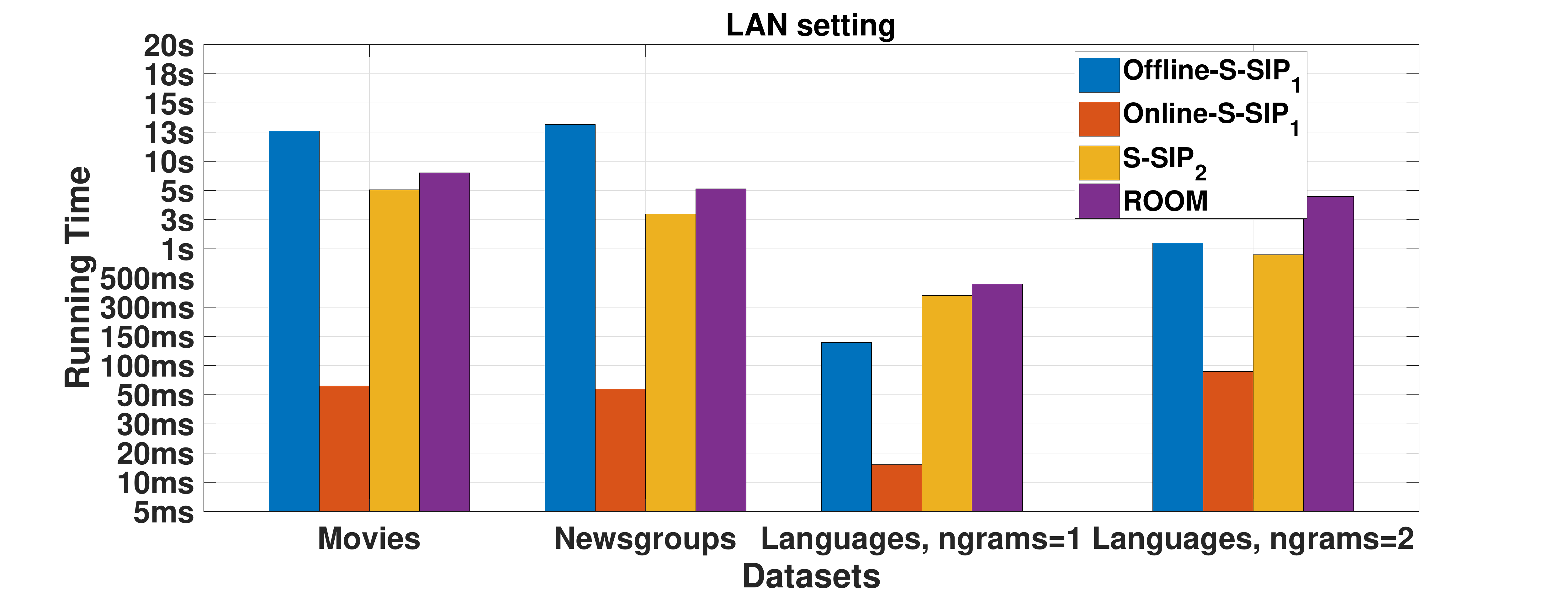}}
 \subfigure[]{\label{6b}\includegraphics[width=0.49\textwidth]{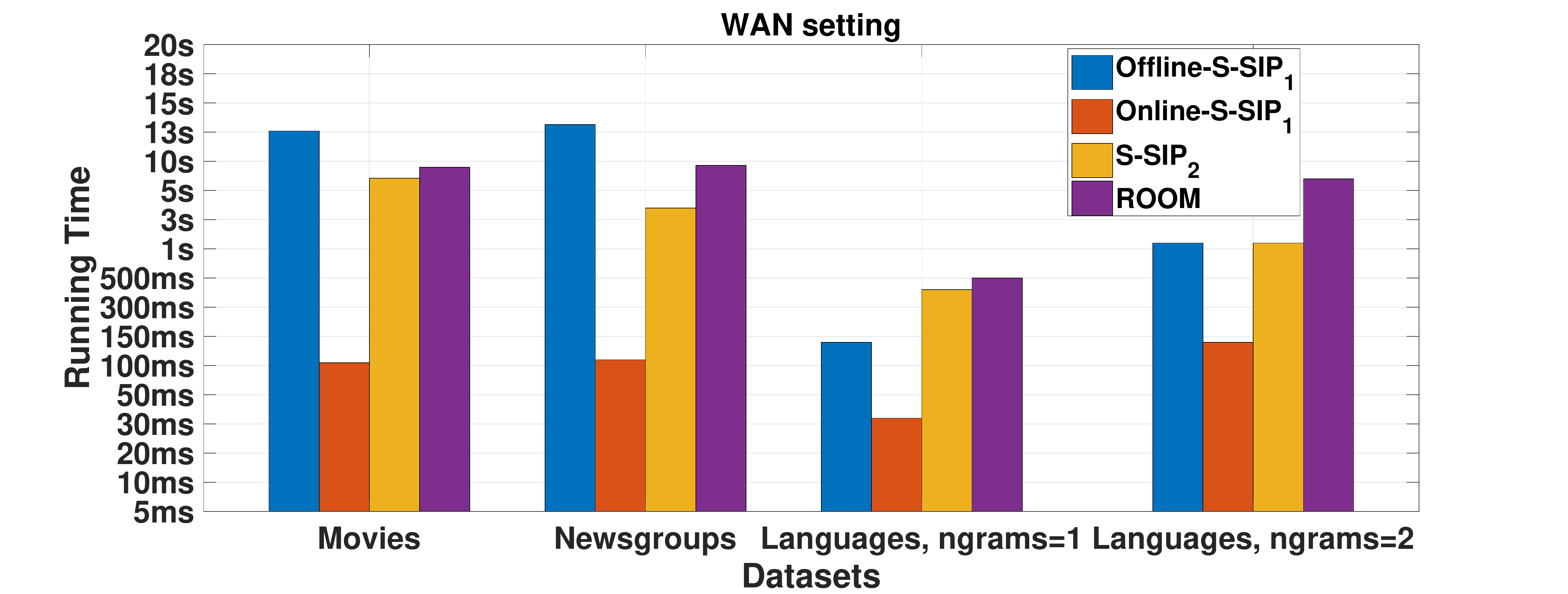}}
  \caption{Running time of logistic regression on different datasets. (a) LAN setting.  (b) WAN setiing. }
  \label{runtime on LG}
  \vspace{-20pt}
\end{figure*}

\begin{figure}
\centering
\includegraphics[width=0.5\textwidth]{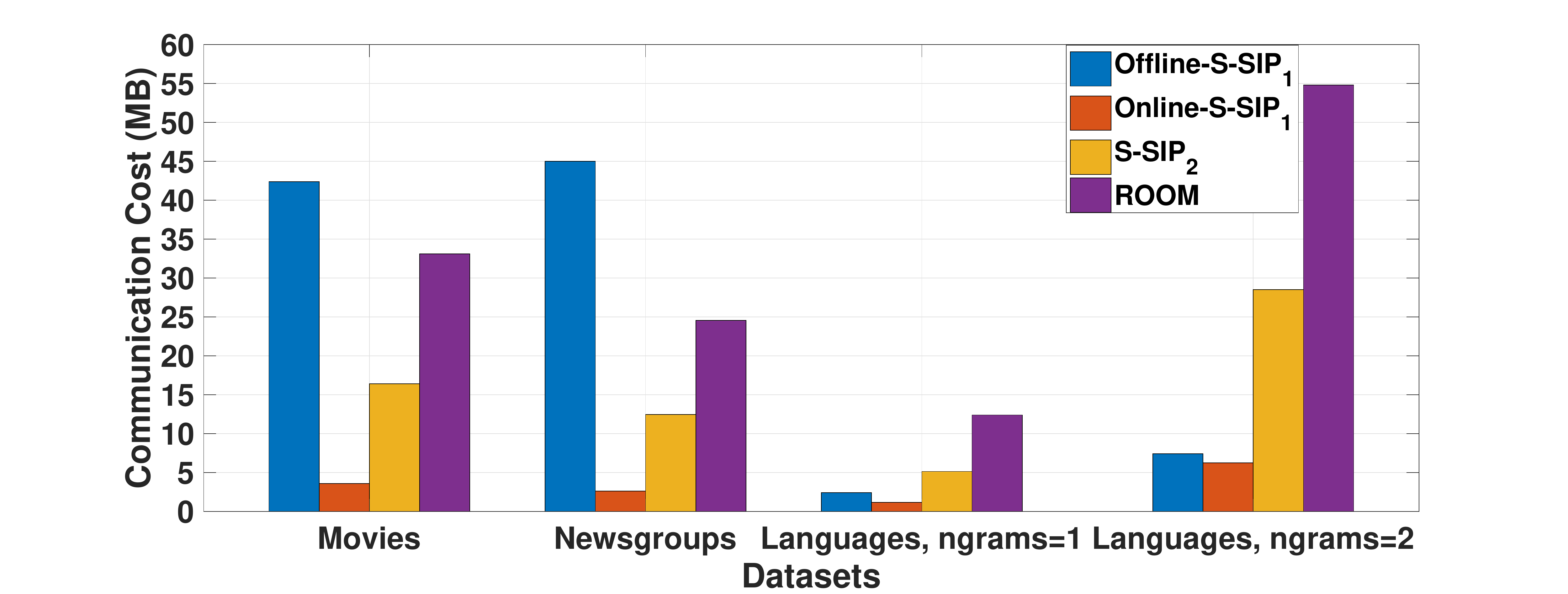}
\caption{Communication cost of logistic regression  on different datasets}
\label{Fig:LG_comm}
 \vspace{-20pt}
\end{figure}

\begin{figure*}[h]
  \centering
  \subfigure[]{\label{10a}\label{Reprint_accuracy}\includegraphics[width=0.49\textwidth]{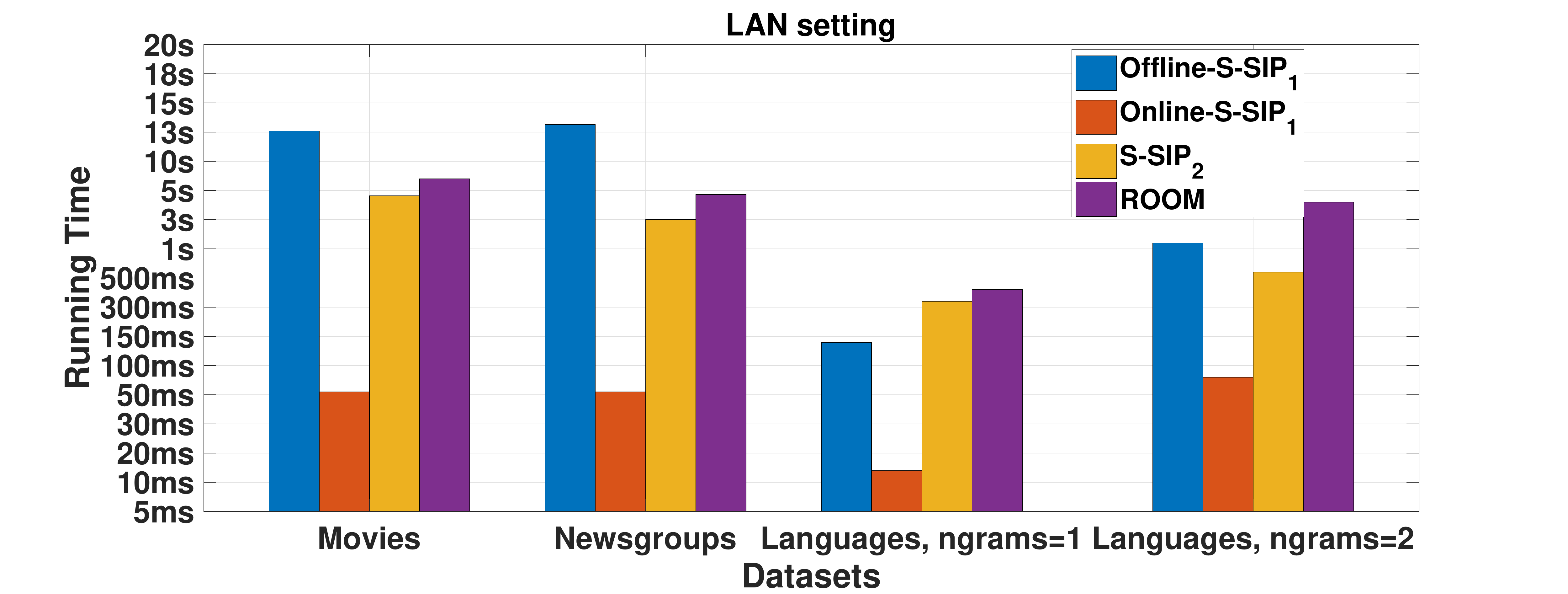}}
 \subfigure[]{\label{10b}\includegraphics[width=0.49\textwidth]{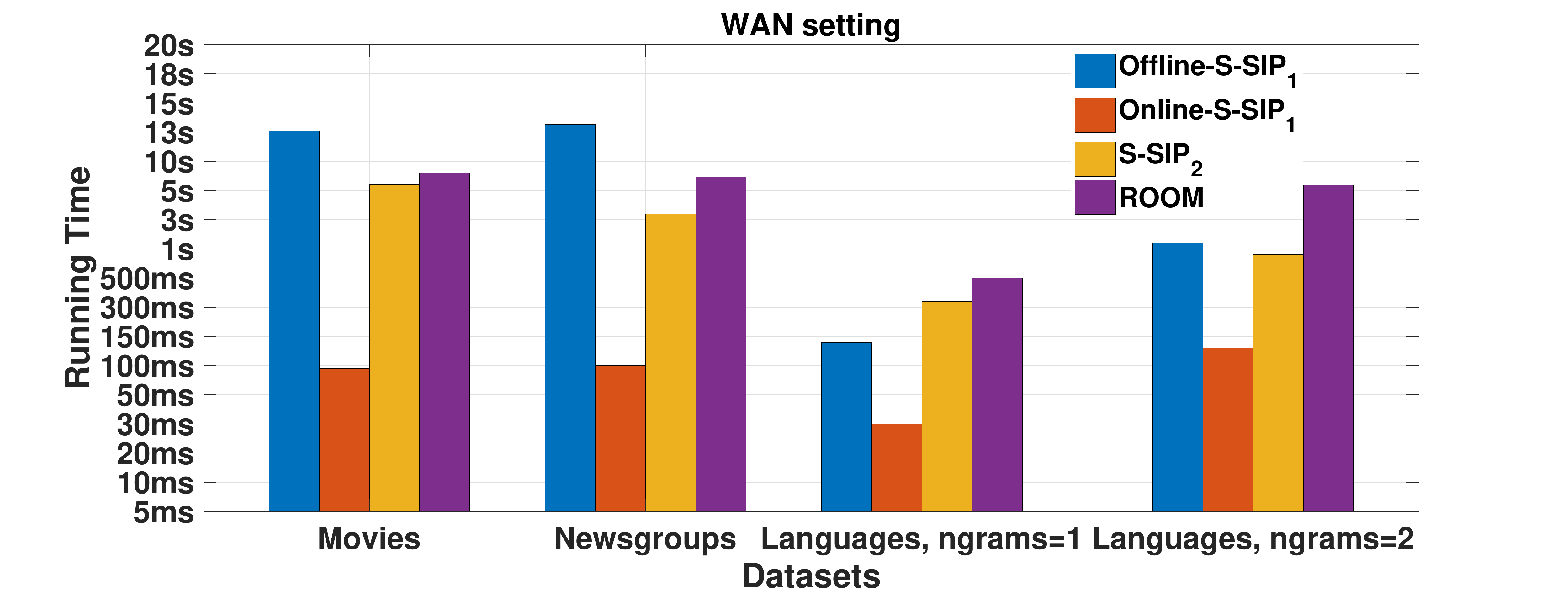}}
  \caption{Running time of naive bayes classification on different datasets. (a) LAN setting.  (b) WAN setiing. }
  \label{runtime on NB}
  \vspace{-15pt}
\end{figure*}

\begin{figure}[htb]
\centering
\includegraphics[width=0.5\textwidth]{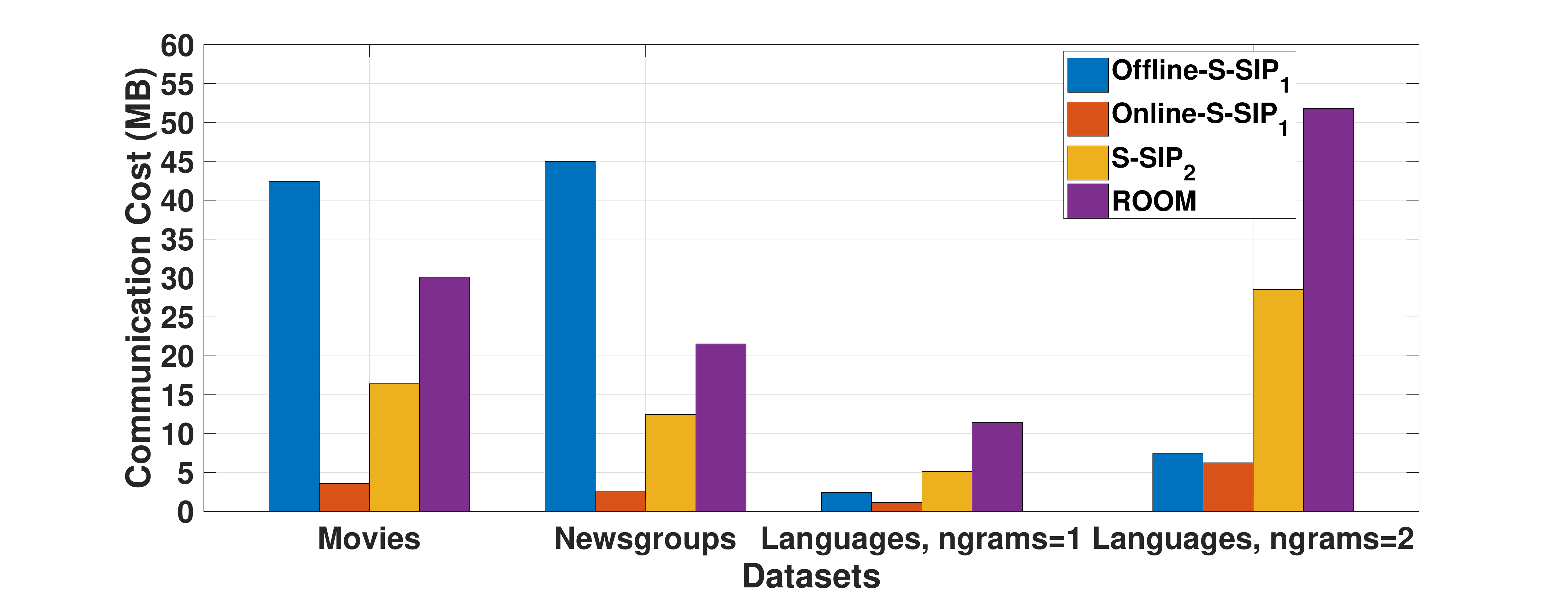}
\caption{Communication cost of naive bayes classification  on different datasets}
\label{Fig:NB_comm}
\vspace{-20pt}
\end{figure}

We first analyze the overhead of each scheme under different dataset sizes. TABLE~\ref{Cost of SIP with Different Size of Datasets} shows the comparison of the computational and communication costs of S-SIP$_1$, S-SIP$_2$ and ROOM under different variables, where the size of the server's database ranges from  $2^{16}$ to $2^{20}$, and the client's from $2^8$ to $2^{16}$. We observe that S-SIP$_1$ relies on heavy overhead for offline, which is linear with the size of the server's database. As a result, the online phase of S-SIP$_1$ is completely independent of the size of the server database, resulting in the best computational speedup of all schemes. For example, when the server holds entries of size $2^{20}$ and the client holds $2^{12}$ entries, performing such a secure inner product operation  S-SIP$_1$ takes only $4.22$ seconds, whileS-SIP$_2$ and ROOM requires $231.7$ and $724.5$ seconds. Moreover, in the online phase, the superiority of the communication overhead saved by of S-SIP$_1$ is evident to increase with the data held by the client. It stems from the fact that the online traffic of S-SIP$_1$  is independent of the server's dataset, while the other two methods are positively related to the server's input.

In the online phase, the computational overhead of of S-SIP$_2$ is higher than that of S-SIP$_1$, since it does not require any precomputation. It is worth noting that when the database held by the server is small, the communication overhead of S-SIP$_2$ is smaller than that of S-SIP$_1$ at certain times. For example, when $n=2^{16}$, $t=2^{16}$, S-SIP$_1$ incurs $1794$ (MB) of traffic while S-SIP$_2$ is about half of S-SIP$_1$. This stems from the Sum-PIR used in S-SIP$_2$, which derives a sublinear communication complexity relative to $t$, while the communication overhead of S-SIP$_1$ increases linearly with $t$. It is clear that S-SIP$_2$ is superior to ROOM in terms of computational and communication overhead. This is mainly due to a series of optimization methods of S-SIP$_2$ for sparse inner product operations, including customized OT protocols to minimize communication overhead, and optimized state-of-the-art PIR technology to accelerate computing. ROOM relies heavily on general-purpose secure multi-party computation to compute the intersection of two datasets, and requires  a large number of Beaver triples to implement multiplication privately. This incurs non-trivial computational and communication costs. As an example, ROOM takes 14396 (MB) and 11598 seconds to complete a secure sparse inner product between one $2^{16}$-dimensional vector and another $2^{20}$-dimensional vector. Conversely, S-SIP$_2$ takes only $3704$ (MB) and 1691.6 seconds, achieving a speedup of 3.89$\times$ and 6.85$\times$, respectively.

\subsection{Performance of Executing $k$-Nearest Neighbors}
\label{Performance of Executing $k$-Nearest Neighbors}

We now discuss the overhead of S-SIP$_1$ and S-SIP$_2$ in performing real ML tasks.  We first consider a $k$-Nearest Neighbor (kNN) task involving a server and a client, where the client holds a labeled database $D$ and the client holds a data $d$ to be classified. In kNN, (a) for each $p\in D$, the client needs to interact with the server to calculate the similarity between the two (the vast majority of the overhead in this process is vector-matrix multiplication); (b) Then, assigning a class $c_d$ to $d$ as the result of a majority vote among the classes of the
$k$ most similar documents according to the similarities computed in step (a) (see work \cite{schoppmann2019make} for details of $k$-Nearest Neighbor). Consistent with work ROOM, we use cosine similarity to calculate the similarity between matrices and vectors, which is essentially a series of  sparse  inner products, which can be easily implemented with S-SIP$_1$ and S-SIP$_2$ as the underlying structure. As for majority voting (step b) to achieve classification for a given input, we follow ROOM's approach with the same secure two-party computation protocol as the carrier.

Fig ~\ref{runtime on KNN} and Fig~\ref{Fig:KNN_comm} show the overhead of each scheme on different datasets. We observe that S-SIP$_1$ incurs heavy offline computation, and the overhead  is linear to the size of the dataset. However, this is done only once, that is, offline one-time pre-computation can support unlimited online queries, resulting in excellent online amortization overhead. For example,  S-SIP$_1$ only needs 12.3(min) and 3.8 (MB) of traffic to classify a single document. Compared with S-SIP$_2$ and ROOM, it saves up to $8 \times$ communication overhead, and brings at least $11 \times$ speedup of computing. It benefits from the design of  S-SIP$_1$ for the online phase, which mainly involves the execution of several efficient OT protocols without computationally intensive homomorphic evaluation.

 Since without any precomputation, the execution cost of S-SIP$_2$ is higher than S-SIP$_1$, but significantly lower than ROOM due to the custom design for SIP. Compared with ROOM, we observe that S-SIP$_2$ achieves at least a $2\times$ improvement in both communication and computing performance. This is due to the customized design of computational SIP in S-SIP$_2$, including partitioning PIR queries and efficient OT executions. On the contrary, ROOM  relies heavily on garbled circuits to execute SIP, which is computationally expensive since even performing simple arithmetic operations requires building circuits containing tens of thousands of ANDs.

\subsection{Performance of Executing Logistic Regression}
\label{Performance of Executing Logistic Regression}

We further discuss the cost comparison of each scheme on logistic regression. We also consider a two-party logistic regression scenario involving a server and a client, where the server holds a classification model and the client holds the input to be classified (see work \cite{schoppmann2019make} for details of  Logistic Regression). At the high-level view,  logistic regression mainly includes two types of computation, one is the inner product operation between the model parameters and the input features, and the other is the execution of an activation function such as Sigmoid. The former can be easily implemented with S-SIP$_1$ and S-SIP$_2$. As for the latter, we follow ROOM's approach, which  performs polynomial fitting on the sigmoid and then encapsulate it in a garbled circuit for private execution.

Fig~\ref{runtime on LG} and Fig~\ref{Fig:LG_comm} show the overhead of each scheme on different datasets. Consistent with the previous one, S-SIP$_1$ shows the best performance in the online phase, although this requires non-trivial precomputation. This makes S-SIP$_1$ ideal for ML scenarios where the server holds small and fixed datasets such as trained ML models. For example,  S-SIP$_1$ only needs 57.1(ms) and 3.6 (MB) of traffic to classify a single document. Compared with S-SIP$_2$ and ROOM, it saves up to $9 \times$ communication overhead, and brings at least $8752  \times$ speedup of computing.  As discussed above, it benefits from the design of  S-SIP$_1$ for the online phase, which mainly involves the execution of several efficient OT protocols without computationally intensive homomorphic evaluation. The execution cost of S-SIP$_2$ is higher than S-SIP$_1$, but significantly lower than ROOM due to the custom design for SIP. We use state-of-the-art PIR technology as the underlying technology and extend it to batch queries mode to reduce overhead through amortization. S-SIP$_2$ enables the overhead of S-SIP to be linear to the client's input (smaller) and logarithmic to the size of the server's dataset. On the contrary, ROOM  relies heavily on garbled circuits to execute SIP, which is computationally expensive.

\subsection{Performance of Executing Naive Bayes Classification}
\label{Performance of Executing Naive Bayes Classification}

We finally consider the naive bayesian classification scenario consisting of a server and a client, where the server holds the database $D$ and the client holds the input features $d$. This scenario includes two processes, (a) one is to calculate the intersection between $d$ and  $D$ on features, (b) and the other is to use Bayesian probability for classification (see work \cite{schoppmann2019make} for details). Since ROOM has no code to implement the latter (i.e, step (b)), in keeping with it, here we only discuss the overhead of securely computing the former.

Fig~\ref{runtime on NB}  and Fig~\ref{Fig:NB_comm} show the overhead of each scheme on different datasets. The results on this experiment are similar to those on the logistic regression task because they perform similar operations under ciphertext. Apparently, S-SIP$_1$ still shows an advantage over the other two methods. The execution cost of S-SIP$_2$ is higher than S-SIP$_1$, but significantly lower than ROOM due to the custom design for SIP.  Compared to ROOM, there are at least a $2\times$ improvement in both communication and computing performance. This is due to the customized design of computational SIP in S-SIP$_2$. Our approach exploits state-of-the-art cryptography tools including garbled Bloom filters and Private Information Retrieval (PIR) as the cornerstone, but carefully fuses them to obtain non-trivial overhead reductions.

\section{Conclusion}
\label{sec:conclusion}
In this paper,  we propose two concrete constructs, S-SIP$_1$ and S-SIP$_2$.  Our approach exploits state-of-the-art cryptography tools including garbled Bloom filters and Private Information Retrieval (PIR) as the cornerstone, but carefully fuses them to obtain non-trivial overhead reductions. We provide a formal security analysis of the proposed constructs and  implement them  into representative machine learning algorithms including k-nearest neighbors, naive Bayes and logistic regression. Compared to the  existing efforts, our method achieves $2$-$50\times$ speedup in runtime and up to $10\times$ reduction in communication. 
In the future, we will focus on designing more efficient optimization strategies to further reduce the computation overhead of our constructions, to make it more suitable for practical applications.
%\clearpage
\ifCLASSOPTIONcaptionsoff
\newpage \fi
\bibliographystyle{IEEEtran}
\bibliography{PPDR}
%\balance
\begin{IEEEbiography}[{\includegraphics[width=1in,height=1.25in,clip,keepaspectratio]{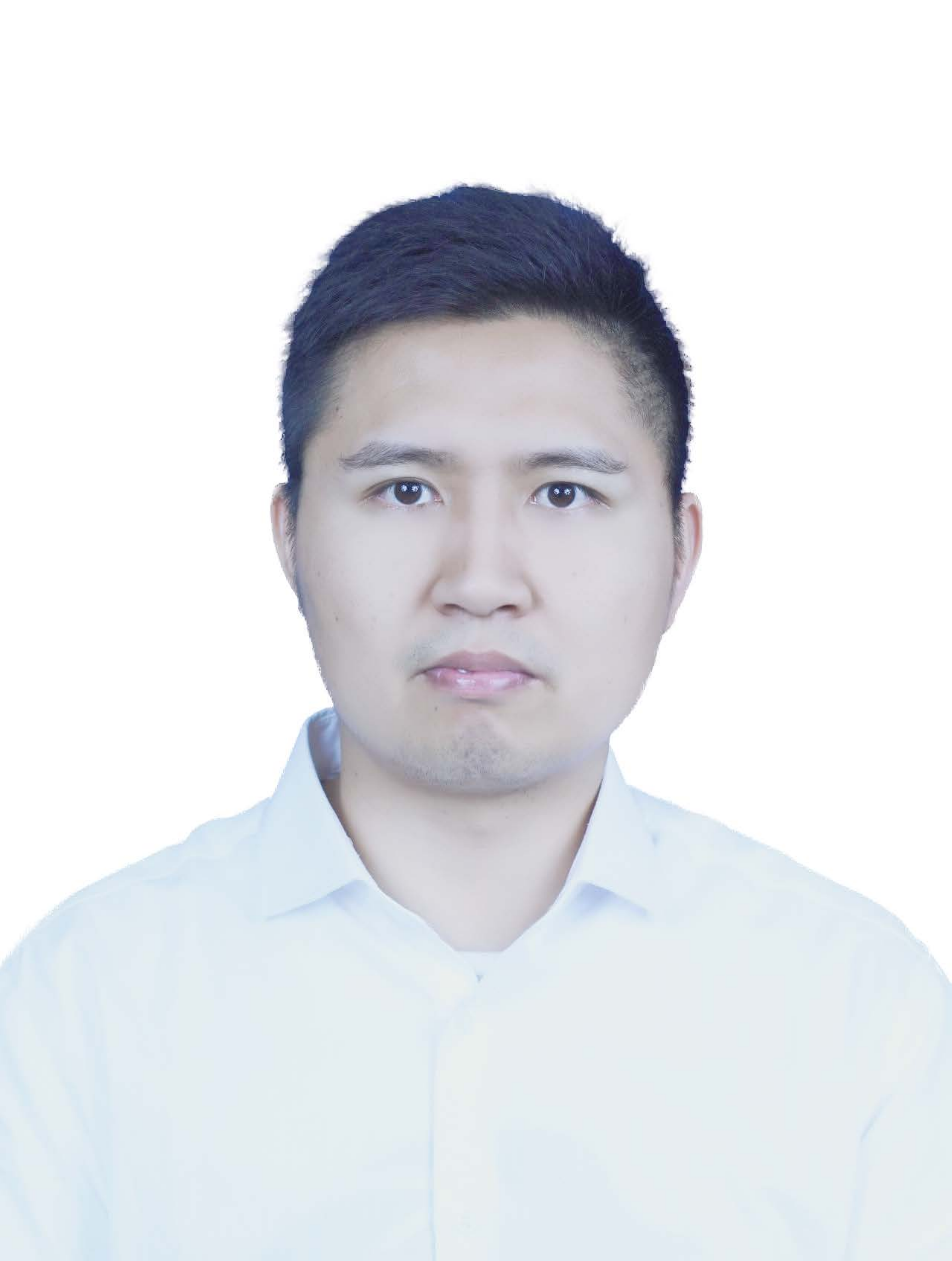}}]{Guowen Xu }
is currently a Research Fellow with Nanyang Technological University, Singapore. He received the Ph.D. degree at 2020 from University of Electronic Science and Technology of China. He has published papers in reputable conferences/journals, including
ACM CCS, NeurIPS, ASIACCS, ACSAC, ESORICS, IEEE TIFS, and IEEE TDSC. His research interests include applied cryptography and  privacy-preserving  Deep Learning.
\end{IEEEbiography}
\vspace{-13 mm}
\begin{IEEEbiography}[{\includegraphics[width=1in,height=1.25in,clip,keepaspectratio]{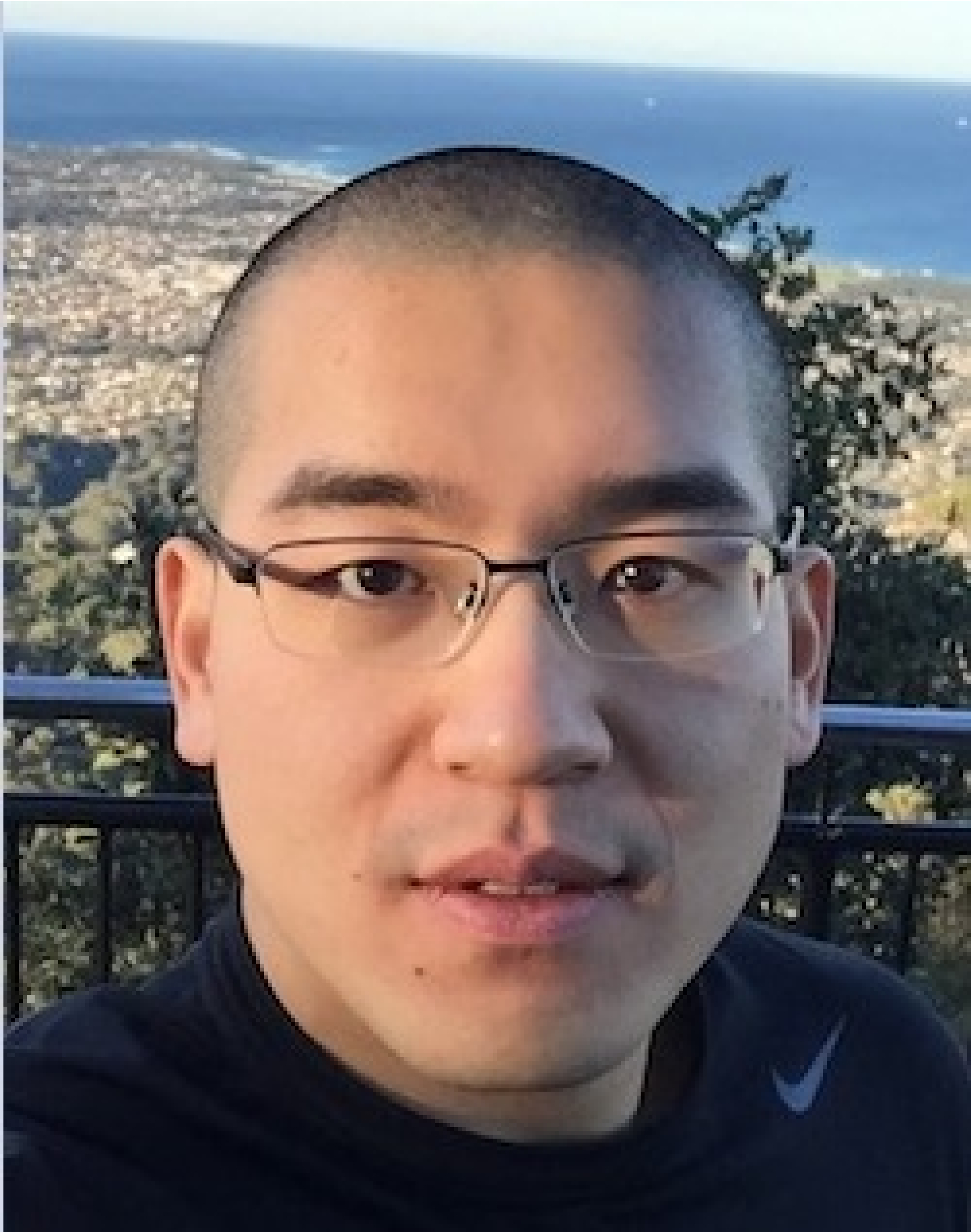}}]{Shengmin Xu }
 is currently an Associate Professor at Fujian Provincial
Key Laboratory of Network Security and Cryptology, College of Computer and Cyber Security, Fujian
Normal University, Fuzhou, China. Previously, he was a Senior Research Engineer with the School of Computing and Information Systems, Singapore
Management University.  His
research interests include cryptography and information security.
\end{IEEEbiography}
\vspace{-13 mm}
\begin{IEEEbiography}[{\includegraphics[width=1in,height=1.25in,clip,keepaspectratio]{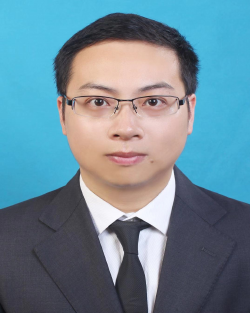}}]{Jianting Ning }
is currently a Professor with the Fujian
Provincial Key Laboratory of Network Security and Cryptology, College of Computer and Cyber Security,
Fujian Normal University, China.  He has published papers in major conferences/journals, such as
ACM CCS, NDSS, ASIACRYPT, ESORICS, ACSAC, IEEE Transactions on
Information Security and Forensics, and IEEE Transactions on Dependable
and Secure Computing. His research interests include applied cryptography
and information security.
\end{IEEEbiography}
\vspace{-13 mm}
\begin{IEEEbiography}[{\includegraphics[width=1in,height=1.25in]{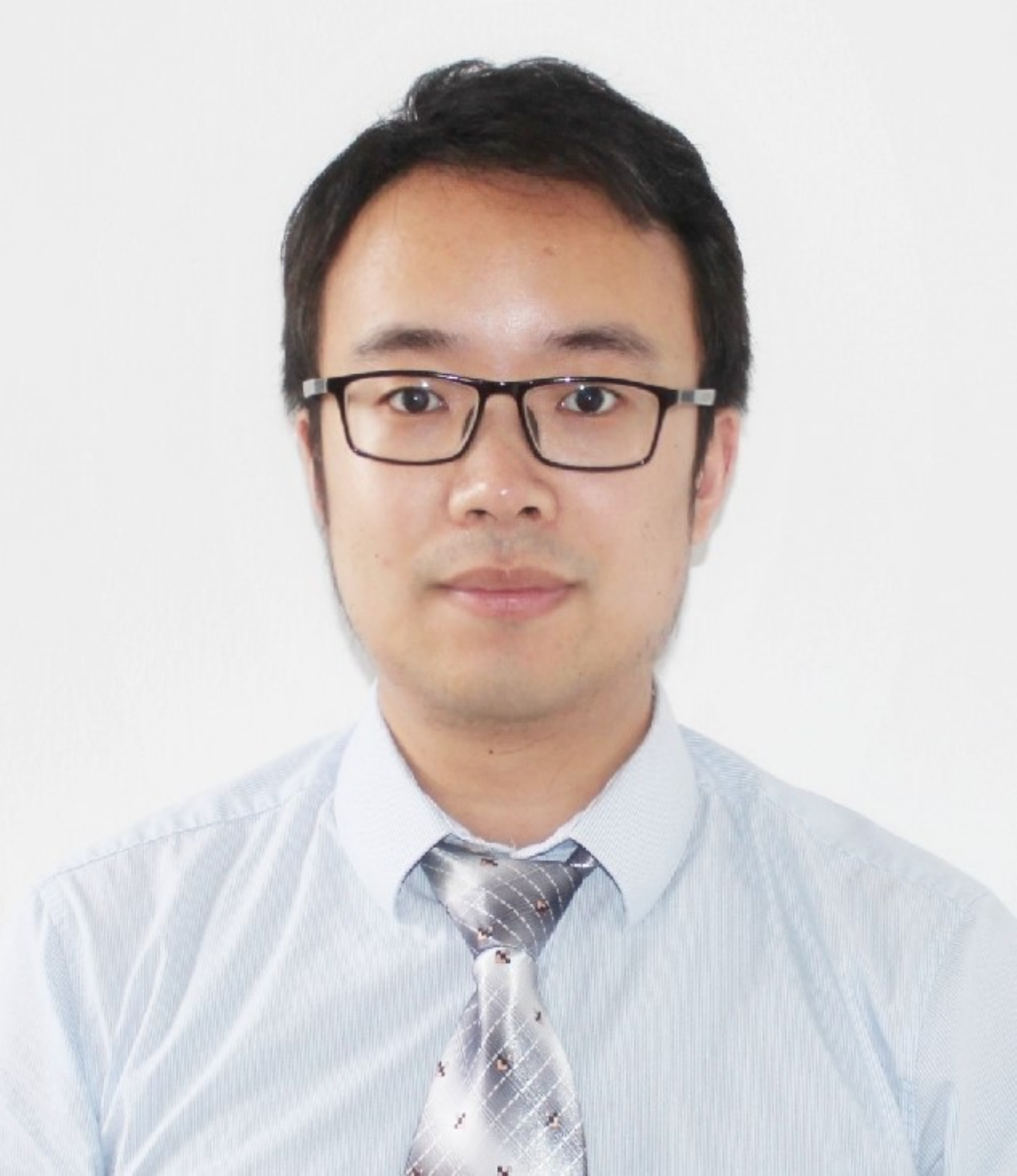}}] {Tianwei Zhang}is an assistant professor in School of Computer Science and Engineering, at Nanyang Technological University. His research focuses on computer system security. He is particularly interested in security threats and defenses in machine learning systems, autonomous systems, computer
architecture and distributed systems. He received his Bachelor's degree at Peking University in 2011, and the Ph.D degree in at Princeton University in 2017.
\end{IEEEbiography}
\vspace{-13 mm}
\begin{IEEEbiography}[{\includegraphics[width=1in,height=1.25in]{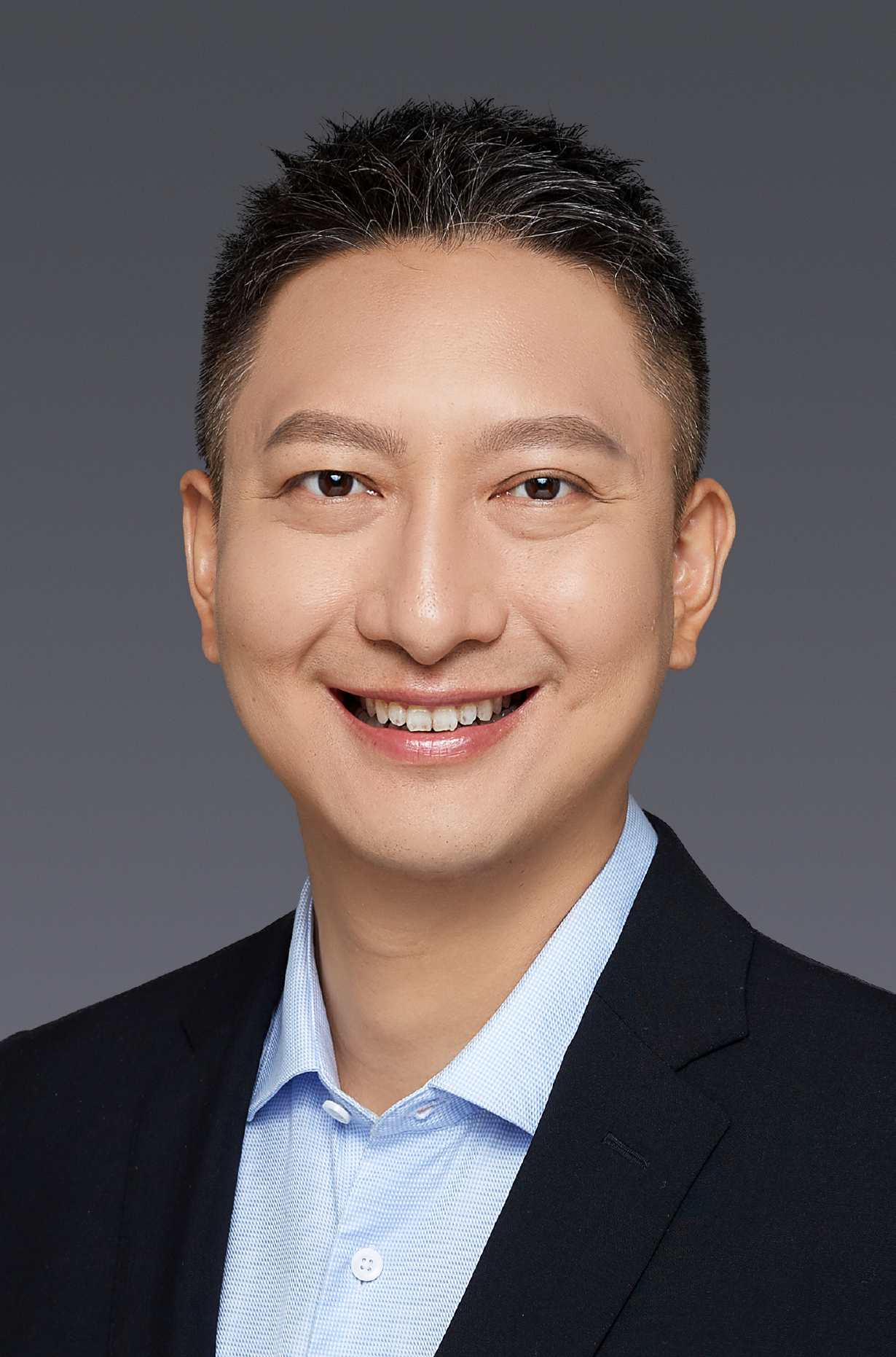}}] {Xinyi Huang}  is currently an Associate Professor at the Thrust of Artificial Intelligence, Information Hub, Hong Kong University of Science and Technology
(Guangzhou), China. His research interests include cryptography and information security.  He is in the Editorial
Board of International Journal of Information Security and SCIENCE CHINA
Information Sciences. He has served as the program/general chair or program
committee member in over 120 international conferences.
\end{IEEEbiography}
\vspace{-13 mm}

\begin{IEEEbiography}[{\includegraphics[width=1in,height=1.9in,clip,keepaspectratio]{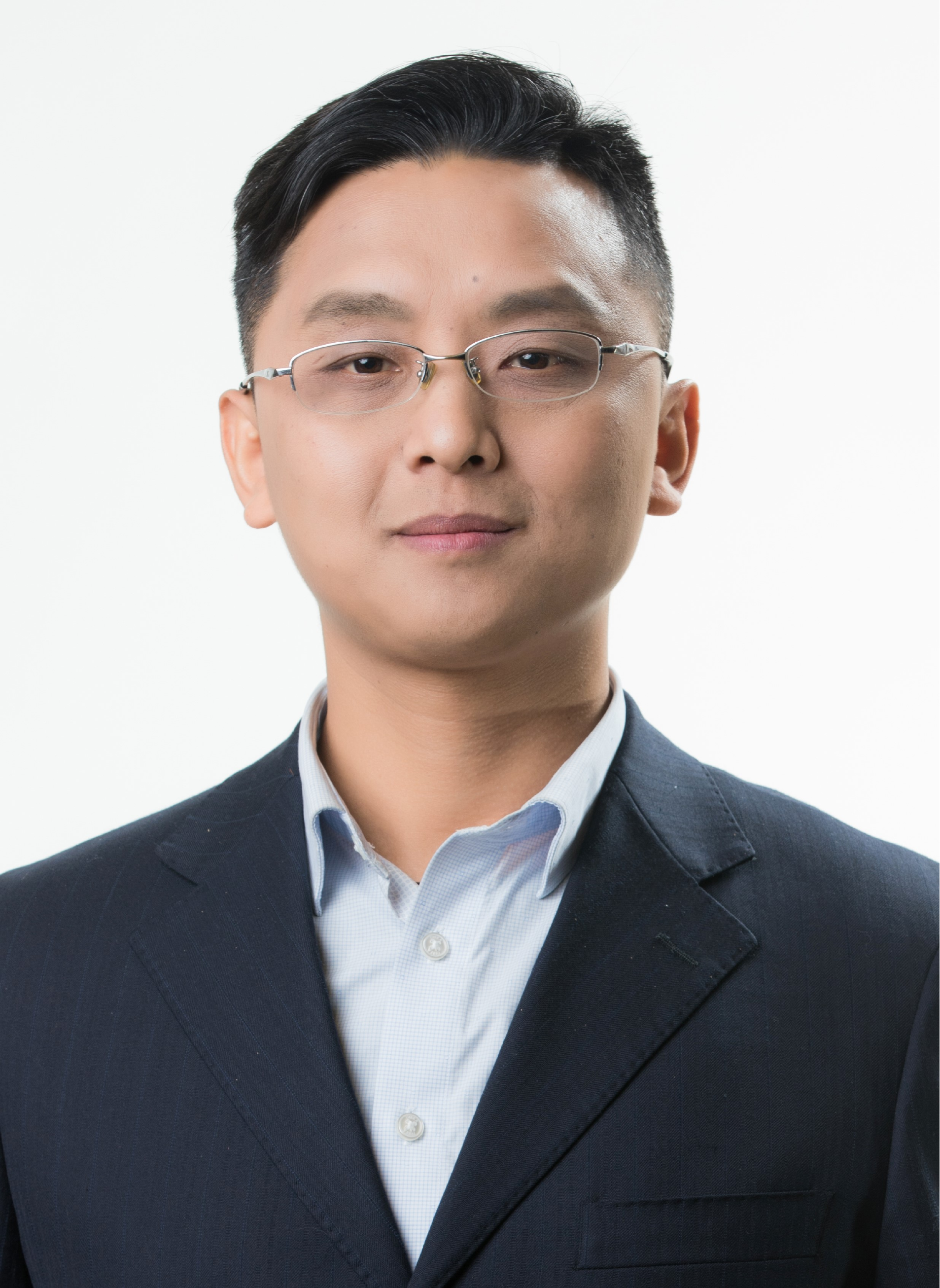}}]{Hongwei Li}
is currently the Head and a Professor at Department of Information Security, School of Computer Science and Engineering, University of Electronic Science and Technology of China.  His research interests include network security and applied cryptography. He is the Senior Member of IEEE, the Distinguished Lecturer of IEEE Vehicular Technology Society.
\end{IEEEbiography}
\vspace{-13 mm}
\begin{IEEEbiography}[{\includegraphics[width=1in,height=1.9in,clip,keepaspectratio]{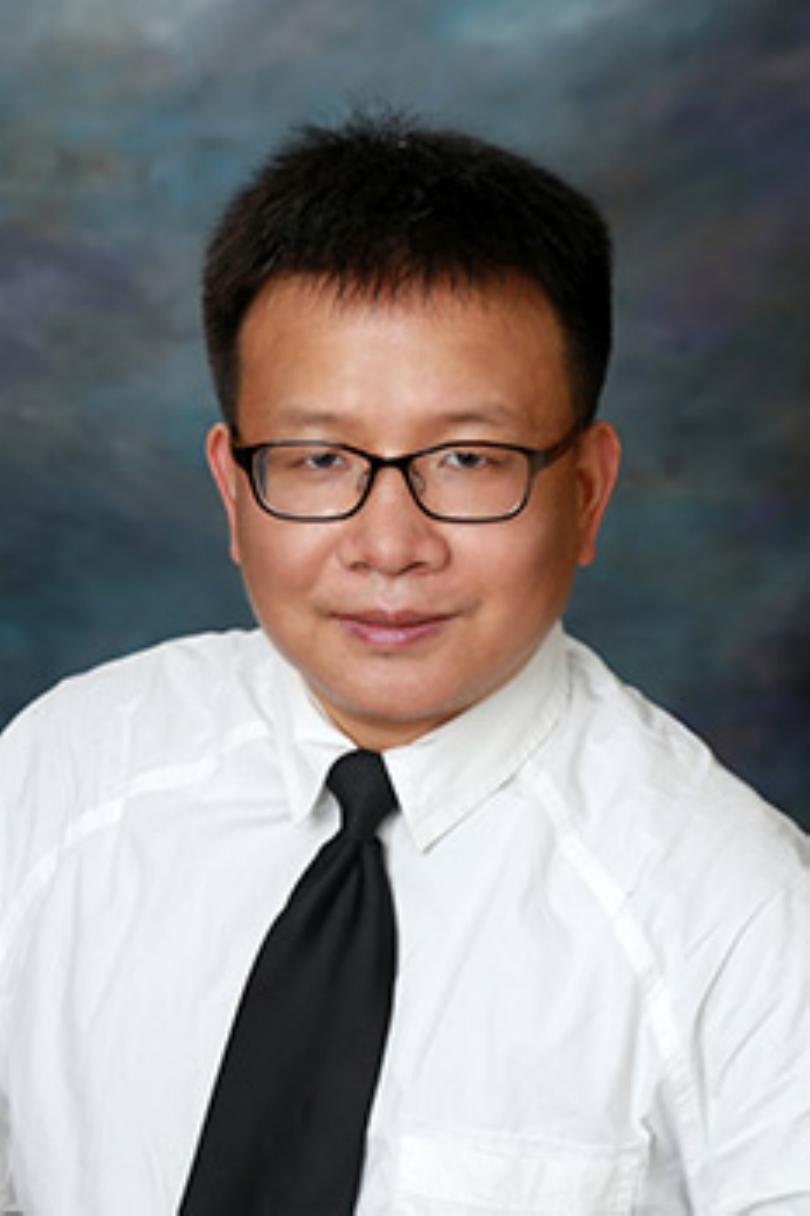}}]{Rongxing Lu}
is currently an associate professor at the Faculty of Computer Science (FCS), University of New
Brunswick (UNB), Canada. He received his PhD degree from
the Department of Electrical \& Computer Engineering, University of Waterloo, Canada, in
2012; and won the 8th IEEE Communications
Society (ComSoc) Asia Pacific (AP) Outstanding Young Researcher Award, in 2013. He is
presently an IEEE Fellow. Dr. Lu currently serves as the Vice-Chair
(Publication) of IEEE ComSoc CIS-TC. Dr. Lu is the Winner of 2016-
17 Excellence in Teaching Award, FCS, UNB.

\end{IEEEbiography}
%\appendices
%%\section*{Appendix}
%\setcounter{section}{0}
%
%\section{Optimized S-SIP$_2$}
%\label{Optimized S-SIP$_2$}
\end{document}